\documentclass[acmsmall]{acmart}

\usepackage[shortlabels]{enumitem}
\usepackage{xifthen}
\usepackage{xstring}
\usepackage{xpatch}
\usepackage{xcolor}
\usepackage{mathtools}
\usepackage[capitalize]{cleveref}
\usepackage{microtype}
\usepackage{tikz}
\usepackage{thm-restate}
\usepackage{subcaption}

\setcitestyle{numbers, sort&compress, square}

\renewcommand{\leq}{\leqslant}
\renewcommand{\geq}{\geqslant}

\renewcommand{\epsilon}{\varepsilon}
\renewcommand{\phi}{\varphi}

\numberwithin{table}{section}
\numberwithin{figure}{section}
\numberwithin{equation}{section}
\makeatletter

\let\@oldstar\*
\newlist{easylist@env}{enumerate}{8}
\setlist[easylist@env,1]{label=\labelitemi}
\setlist[easylist@env,2]{label=\labelitemii}
\setlist[easylist@env,3]{label=\labelitemiii}
\setlist[easylist@env,4]{label=\labelitemiv}
\setlist[easylist@env,5]{label=\labelitemi}
\setlist[easylist@env,6]{label=\labelitemii}
\setlist[easylist@env,7]{label=\labelitemiii}
\setlist[easylist@env,8]{label=\labelitemiv}
\newcounter{easylist@prev}
\newcounter{easylist@this}
\newcounter{easylist@diff}
\renewrobustcmd{\*}{\@ifstar{\easylist@star}{\easylist@nostar}}
\newcommand{\easylist@star}{\addtocounter{easylist@this}{1}\*}
\newcommand{\easylist@nostar}{%
  \def\easylist@type{{easylist@env}}%
  \ifthenelse{\value{easylist@this}>\value{easylist@prev}}{%
    \@ifnextchar{[}{\easylist@settype}{\easylist@checkend}%
  }{%
    \easylist@checkend%
  }
}
\def\easylist@settype[#1]{%
  \def\easylist@type{{easylist@env}[#1]}%
  \easylist@checkend%
}
\newcommand{\easylist@checkend}{%
  \@ifnextchar{/}{%
    \addtocounter{easylist@this}{-1}%
    \expandafter\easylist@go{}\@gobble%
  }{%
    \easylist@go{\item}%
  }%
}
\newcommand{\easylist@go}[1]{%
  \setcounter{easylist@diff}{\value{easylist@this}}%
  \addtocounter{easylist@diff}{-\value{easylist@prev}}%
  \whiledo{\value{easylist@diff}>1}{%
    \addtocounter{easylist@diff}{-1}%
    \begin{easylist@env}%
    \item
  }%
  \ifthenelse{\value{easylist@diff}=1}{%
    \addtocounter{easylist@diff}{-1}%
    \expandafter\begin\easylist@type%
  }{}%
  \whiledo{\value{easylist@diff}<0}{%
    \addtocounter{easylist@diff}{1}%
    \end{easylist@env}%
  }%
  \setcounter{easylist@prev}{\value{easylist@this}}%
  \setcounter{easylist@this}{1}%
  \ifthenelse{\value{easylist@prev}>0}{#1}{}%
}

\def\[#1\]{%
  \begingroup%
  \xpretocmd{\label@in@display}{\yestag}{}{}%
  \noexpandarg\IfSubStr{#1}{\allowdisplaybreaks}{\allowdisplaybreaks}{}%
  \noexpandarg\IfSubStr{#1}{&}{\begin{align*}#1\end{align*}}{%
    \noexpandarg\IfSubStr{#1}{\\}{\begin{gather*}#1\end{gather*}}{%
      \begin{equation*}#1\end{equation*}}}%
  \endgroup%
  \ignorespaces%
}

\newcommand{\yestag@cref@format}{[equation][\arabic{equation}][\thesection]} 
\newcommand*\yestag@make@df@tag@@[1]{%
  \cref@old@make@df@tag@@{#1}%
  \let\cref@old@df@tag\df@tag
  \expandafter\gdef\expandafter\df@tag\expandafter{%
    \cref@old@df@tag
    \def\cref@currentlabel{\yestag@cref@format#1}%
  }%
}
\newcommand*\yestag@make@df@tag@@@[1]{%
  \cref@old@make@df@tag@@@{#1}%
  \let\cref@old@df@tag\df@tag
  \expandafter\gdef\expandafter\df@tag\expandafter{%
    \cref@old@df@tag
    \toks@\@xp{\p@equation{#1}}%
    \edef\cref@currentlabel{\yestag@cref@format\the\toks@}%
  }%
}
\newcommand{\yestag}{%
  \refstepcounter{equation}%
  \begingroup%
    \let\make@df@tag@@\yestag@make@df@tag@@
    \let\make@df@tag@@@\yestag@make@df@tag@@@
  \tag{\theequation}%
  \endgroup%
}

\ExplSyntaxOn
\cs_set_eq:Nc \bers_cref:nn { @cref }
\cs_generate_variant:Nn \bers_cref:nn { nx }
\cs_set_protected:cpn { @cref } #1 #2
 {
  \seq_set_split:Nnn \l_bers_cref_seq { , } { #2 }
  \bers_cref:nx { #1 } { \seq_use:Nn \l_bers_cref_seq { , } }
 }
\seq_new:N \l_bers_cref_seq
\ExplSyntaxOff

\renewcommand{\autoref}{\cref}

\newcommand{\cCrefname}[3]{\crefname{#1}{#2}{#3}\Crefname{#1}{#2}{#3}}

\newcommand{\crefShorten}{%
}
\newcommand{\crefShortenAdd}[4]{\appto{\crefShorten}{#1{#2}{#3}{#4}}}

\crefname{equation}{\@gobble}{\@gobble}
\Crefname{equation}{Equation}{Equations}
\crefShortenAdd\Crefname{equation}{Eq.}{Eqs.}

\crefShortenAdd\cCrefname{chapter}{Ch.}{Chs.}
\crefShortenAdd\cCrefname{section}{\textsection}{\textsection\textsection}
\crefShortenAdd\cCrefname{subsection}{\textsection}{\textsection\textsection}
\crefShortenAdd\cCrefname{subsubsection}{\textsection}{\textsection\textsection}
\crefShortenAdd\cCrefname{paragraph}{\textparagraph}{\textparagraph\textparagraph}
\crefShortenAdd\cCrefname{subparagraph}{\textparagraph}{\textparagraph\textparagraph}
\crefShortenAdd\cCrefname{appendix}{Appx.}{Appxs.}

\cCrefname{table}{Table}{Tables}
\crefShortenAdd\cCrefname{table}{Tab.}{Tabs.}
\cCrefname{figure}{Figure}{Figures}
\crefShortenAdd\cCrefname{figure}{Fig.}{Figs.}

\crefname{footnote}{footnote}{footnotes}
\Crefname{footnote}{Footnote}{Footnotes}

\newcommand{\newrestatable}[1]{%
  \csdef{restatable:#1}{\restatable@begin{}{#1}}
  \csdef{restatable*:#1}{\restatable@begin{*}{#1}}
  \csdef{endrestatable:#1}{\csuse{endrestatable}}
  \csdef{endrestatable*:#1}{\csuse{endrestatable*}}
}

\theoremstyle{acmplain}

\newtheorem{theorem}{Theorem}[section]
\newrestatable{theorem}
\cCrefname{theorem}{Theorem}{Theorems}
\crefShortenAdd\cCrefname{theorem}{Thm.}{Thms.}

\newrestatable{lemma}
\cCrefname{lemma}{Lemma}{Lemmas}
\crefShortenAdd\cCrefname{lemma}{Lem.}{Lems.}

\newrestatable{proposition}
\cCrefname{proposition}{Proposition}{Propositions}
\crefShortenAdd\cCrefname{proposition}{Prop.}{Props.}

\newrestatable{corollary}
\cCrefname{corollary}{Corollary}{Corollaries}
\crefShortenAdd\cCrefname{corollary}{Cor.}{Cors.}

\newrestatable{conjecture}
\cCrefname{conjecture}{Conjecture}{Conjectures}
\crefShortenAdd\cCrefname{conjecture}{Conj.}{Conjs.}

\theoremstyle{acmdefinition}

\newtheorem{definition}[theorem]{Definition}
\newrestatable{definition}
\cCrefname{definition}{Definition}{Definitions}

\crefShortenAdd\cCrefname{definition}{Def.}{Defs.}

\newrestatable{example}
\cCrefname{example}{Example}{Examples}
\crefShortenAdd\cCrefname{example}{Ex.}{Exs.}

\newrestatable{assumption}
\cCrefname{assumption}{Assumption}{Assumptions}
\crefShortenAdd\cCrefname{assumption}{Asm.}{Asms.}

\newrestatable{claim}

\newrestatable{remark}
\cCrefname{remark}{Remark}{Remarks}
\crefShortenAdd\cCrefname{remark}{Rmk.}{Rmks.}

\newtheorem{case}{Case}
\cCrefname{case}{Case}{Cases}

\newcommand{\@thmheadcase}[3]{\thmname{#1}\thmnumber{ #2}\thmnote{: #3}}
\pretocmd{\case}{\let\thmhead\@thmheadcase}{}{}
\pretocmd{\proof}{\setcounter{case}{0}}{}{}

\def\restatable@begin@aux#1#2[#3]#4#5{%
  \ifx#4\label
  \else
    \PackageError{prelude}
      {Missing label in restatable theorem environment (restatable#1:#2)}
      {The first thing in the environment should be a label.
        You can use the label name with the restate command to restate the theorem.}
  \fi
  \csuse{restatable#1}[#3]{#2}{#5}\label{#5}
}
\newcommand{\restatable@begin}[2]{\@ifnextchar[{\restatable@begin@aux{#1}{#2}}{\restatable@begin@aux{#1}{#2}[]}}
\newcommand{\restate}{%
    \phantomsection
    \@ifstar{\restate@star@checkref}{\restate@nostar@checkref}%
}
\newcommand{\restate@star@checkref}{%
  \@ifnextchar{\ref}{\expandafter\restate@star\@gobble}{\restate@star}}
\newcommand{\restate@nostar@checkref}{%
  \@ifnextchar{\ref}{\expandafter\restate@nostar\@gobble}{\restate@star}}
\newcommand{\restate@star}[1]{%
  \ifcsdef{#1}{\csuse{#1}*}{\PackageError{prelude}
    {Attempted to restate an undeclared theorem/lemma/etc.}
    {You need to use a restatable:theorem environment to declare a theorem before restating it}}}
\newcommand{\restate@nostar}[1]{%
  \ifcsdef{#1}{\csuse{#1}}{\PackageError{prelude}
    {Attempted to restate an undeclared theorem/lemma/etc.}
    {You need to use a restatable:theorem environment to declare a theorem before restating it}}}

\xpatchcmd{\thmt@restatable}
  {\let\label=\@gobble}
  {\let\label=\gobbled@cleveref@label}
  {}{}
\newcommand\gobbled@cleveref@label[2][]{}

\newcommand{\mvert}{|}
\newcommand{\given}{\nonscript\;\mvert\nonscript\;\mathopen{}}
\newcommand{\resizegiven}{%
  \renewcommand{\given}{\nonscript\;\delimsize\mvert\nonscript\;\mathopen{}}}
\DeclarePairedDelimiterX\gp[1](){\resizegiven #1}
\DeclarePairedDelimiterX\sqgp[1][]{\resizegiven #1}
\DeclarePairedDelimiterX\curlgp[1]\{\}{\resizegiven #1}
\DeclarePairedDelimiterX\vgp[1]\lvert\rvert{\resizegiven #1}
\DeclarePairedDelimiterX\Vgp[1]\lVert\rVert{\resizegiven #1}
\DeclarePairedDelimiterX\angp[1]\langle\rangle{\resizegiven #1}

\newcommand{\afterScripts}[1]{%
  \@ifnextchar_{\afterScripts@sub{#1}}{%
    \@ifnextchar^{\afterScripts@sup{#1}}{%
      #1}}}
\def\afterScripts@sub#1_#2{_{#2}\afterScripts{#1}}
\def\afterScripts@sup#1^#2{^{#2}\afterScripts{#1}}

\undef\P
\newrobustcmd{\P}{\mathbb{P}\afterScripts{\sqgp}}
\newrobustcmd{\E}{\mathbb{E}\afterScripts{\sqgp}}

\newtoggle{maybemarginnote@reverse}
\togglefalse{maybemarginnote@reverse}

\newrobustcmd{\gfliptoggle}[1]{%
    \iftoggle{#1}{Q\global\togglefalse{#1}}{R\global\toggletrue{#1}}}

\newcommand{\notemaybemargin}[1]{\@ifstar{\noteyesmargin{#1}}{\notenomargin{#1}}}
\newcommand{\notenomargin}[2]{{#1{#2}}}
\newcommand{\noteyesmargin}[2]{\marginpar{\scriptsize\raggedright #1{#2}}}

\makeatother

\setcopyright{rightsretained}
\acmJournal{POMACS}
\acmYear{2024} \acmVolume{8} \acmNumber{1} \acmArticle{9} \acmMonth{3} \acmPrice{}\acmDOI{10.1145/3639035}

\makeatletter
\gdef\@copyrightpermission{
  \begin{minipage}{0.2\columnwidth}
   \href{https://creativecommons.org/licenses/by/4.0/}{\includegraphics[width=0.90\textwidth]{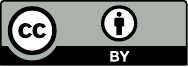}}
  \end{minipage}\hfill
  \begin{minipage}{0.8\columnwidth}
   \href{https://creativecommons.org/licenses/by/4.0/}{This work is licensed under a Creative Commons Attribution International 4.0 License.}
  \end{minipage}
  \vspace{5pt}
}
\makeatother

\newcommand{\lwl}{{\mathrm{LWL}}}
\newcommand{\sitae}{{\textnormal{SITA-E}}}
\newcommand{\sitao}{{\textnormal{SITA-O}}}
\newcommand{\card}{{\mathrm{CARD}}}

\newcommand{\all}{{\mathsf{all}}}
\newcommand{\mgone}{\mathrm{M/G/1}}
\newcommand{\choice}{{\mathsf{choice}}}

\renewcommand{\d}{\mathop{}\!\mathrm{d}}

\newcommand{\noqed}{\renewcommand{\qed}{}}

\newcommand{\I}{\mathcal{I}}

\newcommand{\eqnote}[2]{\stackrel{\mathclap{\textrm{#1}}}{#2}}

\begin{document}

\title{Heavy-Traffic Optimal Size- and State-Aware Dispatching}

\author{Runhan Xie}
\email{runhan\_xie@berkeley.edu}
\affiliation{%
\institution{University of California, Berkeley}
\department{Department of Industrial Engineering and Operations Research}
\city{Berkeley, CA}
\country{USA}
}

\author{Isaac Grosof}
\email{igrosof@cs.cmu.edu}
\affiliation{%
\institution{Carnegie Mellon University}
\department{Computer Science Department}
\city{Pittsburgh, PA}
\country{USA}
}
\affiliation{%
\institution{Georgia Institute of Technology}
\department{School of Industrial and Systems Engineering}
\city{Atlanta, GA}
\country{USA}
}

\author{Ziv Scully}
\email{zivscully@cornell.edu}
\affiliation{%
\institution{Cornell University}
\department{School of Operations Research and Information Engineering}
\city{Ithaca, NY}
\country{USA}
}

\begin{abstract}
Dispatching systems, where arriving jobs are immediately assigned to one of multiple queues, are ubiquitous in computer systems and service systems. A natural and practically relevant model is one in which each queue serves jobs in FCFS (First-Come First-Served) order. We consider the case where the dispatcher is \emph{size-aware}, meaning it learns the size (i.e. service time) of each job as it arrives; and \emph{state-aware}, meaning it always knows the amount of work (i.e. total remaining service time) at each queue. While size- and state-aware dispatching to FCFS queues has been extensively studied, little is known about \emph{optimal} dispatching for the objective of minimizing mean delay. A major obstacle is that no nontrivial lower bound on mean delay is known, even in heavy traffic (i.e. the limit as load approaches capacity). This makes it difficult to prove that any given policy is optimal, or even heavy-traffic optimal.

In this work, we propose the first size- and state-aware dispatching policy that provably minimizes mean delay in heavy traffic. Our policy, called \emph{CARD (Controlled Asymmetry Reduces Delay)}, keeps all but one of the queues short, then routes as few jobs as possible to the one long queue. We prove an upper bound on CARD's mean delay, and we prove the first nontrivial lower bound on the mean delay of any size- and state-aware dispatching policy. Both results apply to any number of servers. Our bounds match in heavy traffic, implying CARD's heavy-traffic optimality. In particular, CARD's heavy-traffic performance improves upon that of LWL (Least Work Left), SITA (Size Interval Task Assignment), and other policies from the literature whose heavy-traffic performance is known.
\end{abstract}



\maketitle

\section{Introduction}
\label{sec:introduction}

Dispatching, or load balancing, is at the heart of many computer systems, service systems, transportation systems, and systems in other domains. In such systems, jobs arrive over time, and each job must be irrevocably sent to one of multiple queues as soon as it arrives. It is common for each queue to be served in First-Come First-Served (FCFS) order.

Motivated by the ubiquity of dispatching, we study a classical problem in dispatching theory:
\begin{quote}
    How should one dispatch to FCFS queues to minimize jobs' mean response time?\/\footnote{%
    A job's \emph{response time} (a.k.a. sojourn time, latency, delay) is the amount of time between its arrival and its completion.}
\end{quote}
We specifically consider \emph{size- and state-aware dispatching}. This means that the dispatcher learns a job's \emph{size}, or service time, when the job arrives; and the dispatcher always knows how much \emph{work}, or total remaining service time, there is at each queue. We make typical stochastic assumptions about the job arrival process, working with M/G arrivals (see \cref{sec:model}).

Despite the extensive literature on dispatching in queueing theory (see \cref{sec:prior_work}), optimal size- and state-aware dispatching is an open problem, as highlighted by \citet{hyytia2022routing}. The problem is a Markov decision process (MDP), so it can in principle be approximately solved numerically \citep{hyytia2023dynamic}. But the numerical approach has two drawbacks. First, the curse of dimensionality makes computation impractical for large numbers of queues. Second, the solution is specific to a particular instance (meaning a given number of queues, job size distribution, and load) and one has to solve the MDP again for a different instance.

\subsection{Our contributions}
In this work, we take the first steps towards developing a theoretical understanding of optimal size- and state-aware dispatching, making two main contributions.
\* We give the first lower bound on the minimum mean response time achievable under any dispatching policy (\cref{thm:rt_lower}).
\* We propose a new dispatching policy, called \emph{CARD (Controlled Asymmetry Reduces Delay)}, and prove an asymptotically tight upper bound on its mean response time (\cref{thm:rt_upper_heavy}). We illustrate CARD in \cref{fig:policy_intro}.
\*/
Our upper and lower bounds match in the heavy-traffic limit as load~$\rho$ approaches~$1$, the maximum load capacity. Specifically, we find an explicit constant~$K$ such that the dominant term of both bounds is $\frac{K}{1 - \rho}$. This makes CARD the first policy to be proven heavy-traffic optimal, aside from the implicitly specified optimal policy. Characterizing the optimal constant~$K$, which was previously unknown, is another contribution of our work.

\subsubsection{How CARD outperforms previous policies}
\begin{figure}
    \begin{minipage}[b]{0.48\linewidth}
        \centering
        \begin{tikzpicture}[xscale=3.95, yscale=3.95, font=\small, every node/.style={inner sep=0,outer sep=2pt}]
\draw[->] (0, 0) -- ++(1, 0) node[right, align=center] {long queue\\work $W_l$};
\draw[->] (0, 0) -- ++(0, 0.65) node[above, align=center] {short queue\\work $W_s$};

\draw[thick, densely dashed, olive] (0, 0.3) -- ++(1, 0);
\draw (0, 0.3) -- ++(-0.03, 0) node[left] {$c$};

\node[align=center] at (0.5, 0.15) {\textit{When $W_s \leq c$:}\\medium $\to$ short};
\node[align=center] at (0.5, 0.45) {\textit{When $W_s > c$:}\\medium $\to$ long};

\begin{scope}[shift={(-0.025, 0)}]
\node at (0.625, -0.125) {\textit{Load division:}};
\draw (-0.15, -0.325) -- ++(0, -0.03) node[below] {\strut \clap{\hphantom{size}size $0$}};
\draw[fill=lime] (-0.15, -0.2) rectangle (0.45, -0.325) node[midway] {\strut small $\to$ short};
\draw (0.45, -0.325) -- ++(0, -0.03) node[below] {\strut size $m_-$};
\draw[fill=yellow] (0.45, -0.2) rectangle (0.8, -0.325) node[midway] {\strut medium};
\draw (0.8, -0.325) -- ++(0, -0.03) node[below] {\strut size $m_+$};
\draw[fill=pink] (0.8, -0.2) rectangle (1.4, -0.325) node[midway] {\strut large $\to$ long};
\draw (1.4, -0.325) -- ++(0, -0.03) node[below] {\strut \clap{size $\infty$\hphantom{size}}};
\end{scope}

\end{tikzpicture}
    \end{minipage}\hfill
    \begin{minipage}[b]{0.48\linewidth}
        \centering
        \includegraphics[width=\linewidth]{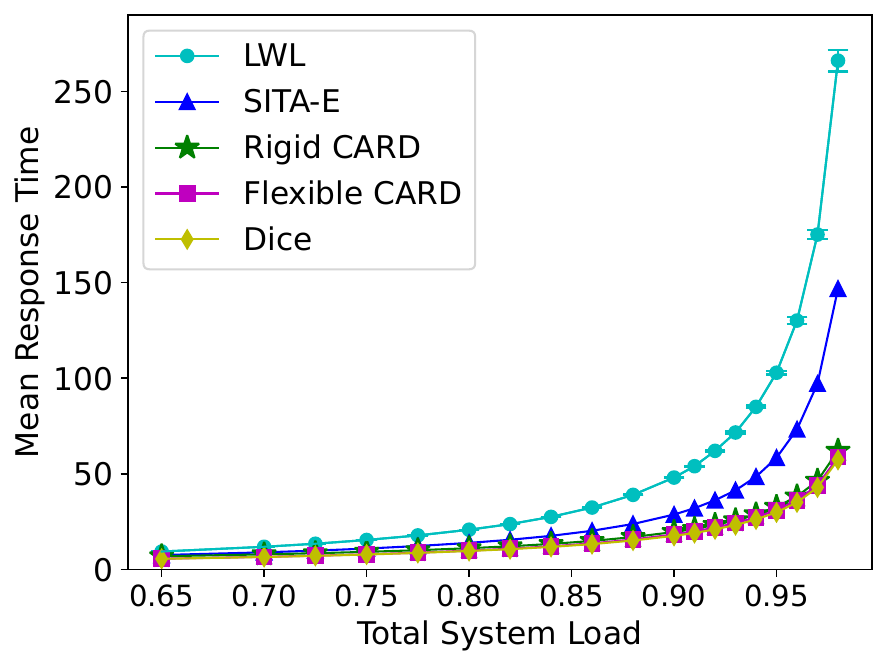}
    \end{minipage}\\[-\baselineskip]
    \begin{minipage}[t]{0.48\linewidth}
        \caption{Sketch of the CARD policy for two servers. Small and large jobs are always dispatched to the short or long server, respectively. Medium jobs are dispatched based on whether $W_s$, the amount of work at the short server, exceeds a threshold~$c$. The size cutoffs $m_-$ and $m_+$ are chosen so that small and large jobs each constitute slightly less than half the load.}
        \label{fig:policy_intro}
    \end{minipage}\hfill
    \begin{minipage}[t]{0.48\linewidth}
        \caption{Mean response time as a function of load for several policies, including two versions of CARD. \emph{Rigid CARD} is the version we theoretically analyze, while \emph{Flexible CARD} is modified slightly to improve empirical performance. The job size distribution has coefficient of variation $\mathsf{cv} = 10$. See \cref{sec:simulation} and \cref{fig:two-server-policies-comparison}(b) for further details.}
        \label{fig:simulation_intro}
    \end{minipage}
\end{figure}

Below, we describe the intuition behind CARD's design in a two-server system. See \cref{fig:policy_intro} for an illustration.

To minimize mean response time, one generally wants to avoid situations where small jobs need to wait behind large jobs. One way to do this is to dedicate one server to small jobs and the other server to large jobs, where the size cutoff between ``small'' and ``large'' is defined such that half the load is due to each size class. This is the approach taken by the SITA (Size Interval Task Assignment) policy \citep{harchol1999choosing, harchol2009surprising}. Under SITA, due to Poisson splitting, the dispatching system reduces to two independent M/G/1 systems. As shown by \citet{harchol2009surprising}, SITA can sometimes perform very well, but it can sometimes be much worse than simple LWL (Least Work Left) dispatching, under which the system behaves like a central-queue M/G/2.

As each of LWL and SITA can sometimes be worse than the other in heavy traffic, one might expect that they can be strictly improved upon. Indeed, in \cref{sec:proofs:suboptimality}, we show that in the two-server case, both LWL and SITA are strictly suboptimal in heavy traffic. But the question remains: where in LWL or SITA's design is there a specific opportunity for improvement?

Our key observation is that the main reason SITA performs poorly is that its ``short server'', namely the queue to which it sends small jobs, can accumulate lots of work. CARD avoids this issue by actively regulating the amount of work at the short server. To do so, CARD creates a third class of ``medium'' jobs, which are on the border between small and large, and sets a threshold which serves as a target amount of work at the short server. Whenever a medium job arrives, CARD dispatches it to the short server if and only if the short server has less work than the threshold. This prevents too much work accumulating in the short server, and it also prevents the short server from unduly idling.

\subsubsection{CARD's performance beyond heavy traffic}
Of course, practical systems rarely operate at loads very near capacity, but our theoretical bounds on CARD's performance are admittedly not tight outside the heavy-traffic regime. As such, we also study CARD in simulation across a wider range of loads. We find empirically that CARD has good performance outside of heavy traffic, but slightly modifying CARD can significantly improves performance. Both the original and modified versions of CARD improve upon traditional heuristics like LWL and SITA, sometimes by an order of magnitude. The modified version is competitive with the Dice policy of \citet{hyytia2022sequential}, the best known heuristic for the size- and state-aware setting. See \cref{fig:simulation_intro} for an example where at high load, CARD achieves reductions of over 75\% relative to LWL and over 50\% compared to SITA.

\subsubsection{Outline}

The remainder of the paper is organized as follows.
\* \cref{sec:prior_work} reviews related work.
\* \cref{sec:model} presents our model and defines the CARD policy.
\* \cref{sec:main_results} states our main results and gives some intuition for why they hold.
\* \cref{sec:lower, sec:stability, sec:upper} prove our results:
    a lower bound on the performance, namely mean response time, of any policy (\cref{sec:lower});
    stability of CARD (\cref{sec:stability}); and
    an upper bound on CARD's performance, which implies its heavy-traffic optimality ($n = 2$ servers in \cref{sec:upper}, general case in \cref{sec:proofs:multi}).
\* \cref{sec:simulation} studies CARD outside of heavy traffic via simulation.
\*/

We note that a preliminary version of this work appeared as a three-page workshop abstract \citep{xie2023reducing}, but it was extremely limited compared to the current version: it treated only the case of two servers and two job sizes, it did not provide any lower bound, and it omitted all proofs.

\subsection{Related Work}
\label{sec:prior_work}

\subsubsection{FCFS dispatching with incomplete information}
Whether a dispatching policy is optimal depends critically on the information available to the dispatcher. When the size of the arriving job is unknown, but server states (e.g. number of jobs at each server, work at each server, etc.) are known (state-aware), depending on the server-state information, Round-Robin (RR) \cite{ephremides1980simple,liu1994optimality,liu1998optimal}, Join-Shortest-Queue (JSQ) \cite{weber1978optimal,winston1977optimality}, and LWL \cite{daley1987certain,foss1980approximation,koole1992optimality,akgun2013partial} are shown to be optimal. The common key idea of these policies is to join the queue with least (or least expected) amount of work.

When only the sizes and the distribution of the arriving jobs are known, SITA is known to be optimal \cite{feng2005optimal}. But this result assumes that the dispatching policy must be entirely static. Recently, it was shown that combining SITA with RR can improve performance \citep{anselmi2019combining,hyytia2020star}, which combines SITA with just a little bit of memory, namely which servers most recently received a job.

Perhaps the closest the SITA line of work gets to size- and state-aware dispatching is the SITA-JSQ policy proposed in \citet{wang2014resource}, in which the dispatcher uses the size of the arriving job and number of jobs at each server to make dispatching decisions. CARD is in some ways similar to SITA-JSQ, particularly the ``multi-band'' variant of CARD introduced in our simulation study (\cref{sec:simulation}). But SITA-JSQ does not actively control the amount of work in each queue, and in particular does not maintain a large imbalance between queues. Our lower bound (\cref{sec:lower}) shows this imbalance is necessary for heavy-traffic optimality.

\subsubsection{FCFS size- and state-aware dispatching}
For size- and state-aware FCFS dispatching, various heuristics have been proposed and studied in simulations. Many of them are based on approximate dynamic programming e.g. \citep{hyytia2012size,hyytia2013lookahead,hyytia2023dynamic}. Another class of policies, called sequential dispatching policies, are introduced in \citep{hyytia2022routing}. Among the sequential dispatching policies, Dice \citep{hyytia2022routing} shows superior performance in simulations and is among the best heuristics that has been developed. In our simulations (\cref{sec:simulation}), Dice often slightly outperforms CARD. However, there is no theoretical analysis so far on the performance of Dice, even in heavy traffic.

\subsubsection{Heavy-traffic Optimality Results}
The aforementioned optimality results are strong in the sense that they either show stochastic ordering optimality on sample paths, or show optimality for any load of jobs. For more complicated policies and systems, characterizing the mean response time for an arbitrary load is a difficult task. Therefore, a large number of works focus on analyzing the heavy-traffic regime and establish optimality therein. One approach is to prove optimality via process limits e.g. \citep{kelly1993dynamic,teh2002critical}. Such approach focuses on the transient regime and interchange of limits are usually not established for analysis in steady state. Another approach is to work directly in the stationary regime and establish heavy-traffic optimality results on mean response times in steady state e.g. \citep{eryilmaz2012asymptotically,zhou2017designing,zhou2018heavy}. However, these optimality results focus on settings where job sizes are unknown, so they do not address our goal of optimal size-aware dispatching.

\subsubsection{Tools and Methodology}
Recently, \citet{eryilmaz2012asymptotically} introduced and popularized a Lyapunov drift-based approach that is applied to study the steady-state performance of
queueing systems in heavy traffic. The approach has been adopted in studying various switches (e.g. \cite{maguluri2016heavy,wang2017heavy,maguluri2018optimal,lange2019heavy,jhunjhunwala2021low,hurtado2022logarithmic,jhunjhunwala2023heavy}), load-balancing algorithms (e.g. \cite{zhou2017designing,zhou2018heavy,zhou2019flexible,weng2020optimal,hurtado2022load,liu2022steady}), and other stochastic models (e.g. wireless scheduling, Stein's method, mean-field models). In some sense, our paper applies drift method to continuous-time continuous state Markov processes. Our use of the Rate Conservation Law \citep{miyazawa1994rate} parallels the use of ``zero drift'' condition in drift analysis. An important step in drift analysis is establishing state-space collapse. We prove a result of this type in \cref{thm:ssc_short}.

\subsubsection{Other Relevant Work}
When scheduling is allowed at the servers, optimal dispatching policies can be very different. When there are multiple parallel SRPT servers, \citet{down2006multi} study a multi-layer dispatching policy and show optimality using a diffusion limit argument. \citet{grosof2019load} develop a dispatching policy, called \emph{Guardrails}, that achieves optimal mean response time in heavy traffic.\footnote{Guardrails is optimal in the sense that the mean response time under Guardrails matches that of a resource-pooled SRPT in heavy-traffic. As \Cref{thm:rt_lower} suggests, dispatching to FCFS servers cannot in general match the performance of a resource-pooled SRPT in heavy-traffic. Therefore, in heavy-traffic, an optimal dispatching policy to SRPT servers generally outperforms an optimal policy to FCFS servers in terms of mean response time.} Both of these prior dispatching policies involve, roughly speaking, balancing work evenly across the multiple SRPT servers. This is in contrast to CARD, which maintains a large imbalance between the multiple FCFS servers. One interpretation is that while SRPT prioritizes jobs at each individual server, CARD prioritizes jobs at the dispatching stage, namely by sending shorter jobs to servers with less work.

In recent years, learning-based dispatching policies have also been studied in literature \citep{samuelsson2018applying,fu2022joint}. CARD involves tuning some parameters that depend on the job size distribution, and we thus assume knowledge of the job size distribution. An interesting question for future work is whether CARD's parameters could be learned online in settings where the job size distribution is unknown.

In the context of scheduling jobs on a single server, when SRPT (Shortest Remaining Processing Time) is shown to be optimal \citep{schrage1968proof}, \citet{chen2021scheduling} show that having two priority classes is sufficient for a good performance in heavy traffic. The heavy-traffic performance of CARD ends up roughly equivalent to the performance of a single-server system with two priority classes. However, we cannot match the performance demonstrated by \citet{chen2021scheduling}: they decrease the fraction of load in the lower-priority class to zero in heavy traffic, whereas CARD's ``lower-priority jobs'', namely those sent to the long server, must constitute a roughly $\frac{1}{n}$ fraction of the load.

\section{System Model and the CARD Policy}
\label{sec:model}

\subsection{Model Description}
\label{sec:model:basics}

We consider a system of $n \geq 2$ identical FCFS (First-Come, First-Served) servers, each of which has its own queue. The system has one central dispatcher, which immediately dispatches jobs to a server when they arrive. We consider M/G job arrivals with (Poisson) arrival rate $\lambda$ and job size distribution~$S$. We assume $\E{S^2} < \infty$. The system load, namely the average rate at which work arrives, is $\rho = \lambda \E{S}$. We assume a server never idles unless there are no jobs present in its queue.

We use the convention that each server completes work at rate~$\frac{1}{n}$, so a job of size~$s$ requires $ns$ time in service.  This convention means the largest possible stability region is $\rho \in [0, 1)$, regardless of the number of servers~$n$. The convention is also convenient when comparing our system's performance to that of a ``resource-pooled'' M/G/1 with the same arrival process and server speed~$1$. We write $\E{W_\mgone}$ for the mean amount of work in such a resource-pooled M/G/1.

We consider \emph{size- and state-aware} dispatching policies. That is, when a job arrives, the dispatcher may use both the job's size and the system state to decide where to dispatch it to. For our purposes, the most important aspect of the system state is the amount of \emph{work} remaining at each server. We write $W_i$ for the amount of work at server~$i$ (but see also \cref{sec:model:card}), $\mathbf{W} = (W_1, \dots, W_n)$ for the vector of work amounts, and $W_\all = \sum_{i = 1}^n W_i$ for the total work. We write $W_i(t)$ or $\mathbf{W}(t)$ when discussing work at a specific time~$t$.

The main metric we consider is \emph{mean response time}. A job's response time is the amount of time between its arrival and completion. Due to our $\frac{1}{n}$ service rate convention, if a job of size~$s$ is dispatched to a server with $w$ work, the job's response time is $n(w + s)$. We write $\E{T_\pi}$ for the mean response time over all jobs (in the usual limiting long-run average sense) under policy~$\pi$.

Purely for simplicity of notation, we assume the job size distribution $S$ has no atoms. This is to ensure that expressions like $\E{S \I(S < m)}$ are continuous functions of~$m$. One can generalize all of our definitions and results to distributions with atoms using a lexicographic ordering trick.\footnote{%
    Have the system assign each job an i.i.d. uniform $U \in [0, 1]$ independent of its size~$S$, and replace comparisons $S < m$ with comparisons $(S, U) \prec (m, v)$ for some $v \in [0, 1]$, where $\prec$ is the lexicographic order. If $\E{S \I((S, U) \prec (m, v))}$ has a jump discontinuity at~$m$, varying $v$ interpolates continuously between the left and right limits.}

\subsection{Defining the CARD Policy}
\label{sec:model:card}

We now introduce our policy, \emph{CARD}, which stands for \emph{Controlled Asymmetry Reduces Delay}. We first present it in the context of $n=2$ servers, then generalize to $n \geq 2$ servers.

\subsubsection{CARD for two servers}
In the $n=2$ case, CARD designates server~1 as the \emph{short server} and server~2 as the \emph{long server}. To emphasize this, when discussing CARD, we write $W_s = W_1$ and $W_\ell = W_2$ for the work at the short and long servers, respectively.

CARD has three threshold parameters to set:
\* The two \emph{size thresholds} $m_-$ and $m_+$, $0 \leq m_- \leq m_+$, divide jobs into small, medium, and large (see below).
\* The \emph{work threshold} $c$, $c \geq m_+$ is, roughly speaking, a target work level for the short server.
\*/
Based on these parameters, CARD dispatches jobs as follows (see also \cref{fig:policy_intro}):
\* A \emph{small job}, namely one with size in $[0, m_-)$, is always dispatched to the short server.
\* A \emph{medium job}, namely one with size in $[m_-, m_+)$, is dispatched depending on $W_s$ at time of arrival. If $W_s \leq c$, it is sent to the short server, and if $W_s > c$, it is sent to the long server.
\* A \emph{large job}, namely one with size in $[m_+, \infty)$, is always dispatched to the long server.
\*/

\subsubsection{Setting CARD's parameters}
There are a range of ways to set $m_-$, $m_+$, and $c$ that yield stability and heavy-traffic optimality. We specify these formally in the statements of \cref{thm:stability, thm:rt_upper_heavy}, but we highlight the key points here (see also \cref{sec:model:quantities}).

The size thresholds $m_-$ and~$m_+$ should be chosen such that small jobs and large jobs are each less than half the load. Formally, we require
\[
    \E{S \I(S < m_-)} < \tfrac{1}{2} \E{S} < \E{S \I(S < m_+)}.
\]
In particular, we have $m_- < m < m_+$, where $m$ is the solution to $\E{S \I(S < m)} = \tfrac{1}{2} \E{S}$. As we show in our lower bound (\cref{thm:rt_lower}), this value~$m$ is in some sense the ideal cutoff between small and large jobs. As such, it is important that in heavy traffic, either $m_- \to m$ or $m_+ \to m$ (or both). We do the former in our upper bound (\cref{thm:rt_upper_heavy}).

The work threshold~$c$ must balance a tradeoff between two concerns. On one hand, we want there to be little work at the short server so that small jobs have low response times. On the other hand, we do not want the short server to run out of work, as excessive idling could increase response times or even cause instability. Roughly speaking, this means setting $c = \Theta\gp[\big]{\gp[\big]{\frac{1}{1 - \rho}}^p}$ for a suitable choice of $p \in (0, 1)$.

It is convenient in our proofs to ensure $c \geq m_+$, so we assume this throughout. It also makes intuitive sense that a single medium job should not bring the short server from empty to above the work threshold. However, this assumption can be easily relaxed at the cost of a little more computation in the proofs.

\subsubsection{Generalizing CARD to any number of servers}
We now generalize the above policy to $n\geq2$ servers. Here we focus on an extension that prioritizes simplicity of analysis while still achieving optimal heavy-traffic performance. In our simulation study (\cref{sec:simulation}), we consider a more complex variant which has better performance at practical loads.

The basic idea of $n$-server CARD is to reduce to the two-server case. We use the same three parameters $m_-$, $m_+$, and~$c$, and we define small, medium, and large jobs in the same way. The only difference is that instead of one short and one long server, we use $n - 1$ short servers $1, \dots, n - 1$ and a single long server~$n$. We thus write $W_{s_i} = W_i$ and $W_\ell = W_n$ when discussing $n$-server CARD. Abusing notation slightly, we write simply $W_s$ when discussing a generic short server whose index is not important.
Jobs are dispatched as follows:
\* A small job is always dispatched to a uniformly random short server.
\* A medium job is dispatched as follows. The dispatcher selects a uniformly random short server $N \in \{1, \dots, n - 1\}$ and inspects its amount of work~$W_{s_N}$. If $W_{s_N} \leq c$, the job is dispatched to the chosen short server~$N$, and if $W_{s_N} > c$, it is dispatched to the long server.
\* A large job is always dispatched to the long server.
\*/

Another way to view $n$-server CARD is in the following distributed manner. Suppose that instead of one dispatcher, we have $n - 1$ independent ``subdispatchers'', each associated with a short server, and suppose that all jobs arrive at a uniformly random dispatcher. Then $n$-server CARD is the result of each of the subdispatchers using two-server CARD, except they all share the same long server.

The way we set the parameters of $n$-server CARD is essentially the same as how we set the parameters of two-server CARD. The only difference is that instead of wanting small and large jobs to both have less than half the load, we want small jobs to be less than a $1 - \frac{1}{n}$ fraction of the load, and we want large jobs to be less than a $\frac{1}{n}$ fraction of the load. We therefore set
\[
    \E{S \I(S < m_-)} < \gp*{1 - \frac{1}{n}} \E{S} < \E{S \I(S < m_+)}.
\]
This means $m_- < m < m_+$, where now $m$ is the solution to $\E{S \I(S < m)} = \gp[\big]{1 - \frac{1}{n}} \E{S}$.

\subsection{Key Definitions for Main Results and Analysis}
\label{sec:model:quantities}

We state our main results and perform our analysis in terms of the following quantities.

\subsubsection{Drift-related quantities}
The following quantities are related to characterizing \emph{drifts}, which are the average rates at which work increases or decreases in various situations.
\* Let $\epsilon = 1 - \rho$. If both servers are busy, then $W_\all$ has drift $-\epsilon$.
\* Let $\rho_s$, $\rho_m$, and $\rho_\ell$ be the loads due to small, medium, and large jobs, respectively:
\[\rho_s=\lambda\E{S\I(S < m_-)},\qquad\rho_m=\lambda\E{S\I(m_- \leq S < m_+)},\qquad\rho_\ell=\lambda\E{S\I(S\geq m_+)}.\]
\* Let $\alpha$ and $\beta$ be the following quantities related to the drift of $W_s$:
\[\alpha=\frac{1}{n}-\frac{1}{n-1}\rho_s,\qquad\beta=\frac{1}{n-1}(\rho_s+\rho_m)-\frac{1}{n}.\]
If $W_s > c$, then $W_s$ has drift $-\alpha$, and if $0<W_s \leq c$, then $W_s$ has drift~$+\beta$.
\* Let $\delta \in (0, \epsilon]$ be a bound on the probability the short server is idle, i.e. $\P{W_s = 0} \leq \delta$. We show how to set CARD's parameters to achieve this bound in \cref{thm:stability}(a).
\*/

To specify CARD's $m_-$ and $m_+$ parameters, it suffices to specify $\alpha$ and~$\beta$: these determine $\rho_s$ and~$\rho_m$, which in turn determine  $m_-$ and~$m_+$. Moreover, for any given~$\beta$, we show in \cref{thm:stability} how to set CARD's $c$ parameter to achieve $\P{W_s = 0} \leq \delta$. As such:
\begin{quote}
    Instead of specifying $m_-$, $m_+$, and $c$ directly, we specify $\alpha$, $\beta$, and~$\delta$.
\end{quote}
In particular, \cref{thm:rt_upper_heavy} specifies how $\alpha$, $\beta$, and~$\delta$ should scale as functions of~$\epsilon$.

\subsubsection{Heavy traffic}
Our main results consider the $\epsilon \downarrow 0$ limit, which we call the \emph{heavy-traffic} regime. This is equivalent to $\lambda \uparrow 1/\E{S}$. In particular, we leave the number of servers fixed.

Underlying our results are explicit bounds that hold even outside the limiting regime (see e.g. \cref{thm:rt_upper_explicit}). Because of our focus on heavy traffic, we assume for convenience that $\epsilon < \frac{1}{n}$. In particular, this ensures we can set $\beta > 0$, which ensures that $W_s$ always drifts towards~$c$. The case where $\epsilon > \frac{1}{n}$ and $\beta < 0$ is less interesting, as then both $W_s$ and $W_\ell$ always drift towards~$0$.

\subsubsection{Performance-related quantities}
The following quantities are used in our response time bounds (\cref{thm:rt_lower, thm:rt_upper_heavy}). Define $K_\card$ and $m$ such that
\[
    \label{eq:Kcard}
    K_\card = \frac{\E{S}}{\E{S \given S \geq m}} = n \P{S \geq m}.
\]
This characterization of $m$ is equivalent to the aforementioned $\E{S \I(S < m)} = \gp[\big]{1 - \frac{1}{n}} \E{S}$. In \cref{thm:rt_upper_heavy}, we show that, roughly speaking, $\E{T_\card} \approx K_\card \E{W_\mgone}$, where
\[
    \E{W_\mgone} = \frac{\lambda \E{S^2}}{2 \epsilon}
\]
is the mean work in a resource-pooled M/G/1 (\cref{sec:model:basics}).

\section{Main Results and Key Ideas}
\label{sec:main_results}

We now present our main results, followed by some intuition for why they hold. See \cref{sec:lower, sec:stability, sec:upper} for the proofs, with some details deferred to \cref{sec:proofs}.

Our first result is a lower bound on the mean response time for any dispatching policy.

\begin{restatable:theorem}
    \label{thm:rt_lower}
    Under any dispatching policy~$\pi$ and for any $\epsilon\in(0,1)$,
    \[
        \E{T_\pi} \geq K_\card \E{W_\mgone} - \frac{(n - 1) \E{S^2}}{2 m} + n \E{S}.
    \]
\end{restatable:theorem}

\begin{proof}
    See \cref{sec:lower}.
    \noqed
\end{proof}

The rest of our results are about CARD: stability for all $\epsilon>0$, and heavy-traffic optimality as $\epsilon \downarrow 0$. Both results are stated as sufficient conditions on CARD's parameters under which it achieves the corresponding property. See \cref{sec:model:card, sec:model:quantities} for descriptions of and notation for CARD's parameters.

\begin{restatable:theorem}
    \label{thm:stability}
    Let $\delta > 0$, and consider CARD with threshold
    \[
        c = \frac{n(n-1)m_+}{\beta} \log \frac{n+1}{n \beta \delta}.
    \]
    Then,
    \*[(a)] Each short server satisfies $\P{W_s = 0} \leq \delta$.
    \* If $\delta < \frac{n}{n - 1}\epsilon$, then the system is stable. Specifically, the set $\{(0, \ldots, 0)\}$ is positive recurrent for the process $\mathbf{W}(t) = (W_1(t),\ldots, W_n(t))$.
    \*/
\end{restatable:theorem}

\begin{proof}
    See \cref{sec:stability, sec:proofs:stability}.
    \noqed
\end{proof}

\begin{restatable:theorem}
    \label{thm:rt_upper_heavy}
    For any fixed number of servers $n \geq 2$, if CARD's parameters are set such that
    \[
        \alpha = \Theta(1), \qquad
        \beta = \Theta\gp[\bigg]{\epsilon^{1/3} \gp[\bigg]{\log\frac{1}{\epsilon}}^{2/3}}, \qquad
        \text{and} \qquad
        c = \frac{n(n-1)m_+}{\beta} \log \frac{n+1}{n \beta \delta},
    \]
    in the $\epsilon\downarrow0$ limit, then CARD achieves mean response time bounded by
    \[
        \E{T_{\card}} \leq K_\card \E{W_\mgone} + O\gp[\bigg]{\gp[\bigg]{\frac{1}{\epsilon} \log\frac{1}{\epsilon}}^{1/3}}.
    \]
    In particular, CARD is heavy-traffic optimal: $\limsup_{\epsilon\downarrow0}\frac{\E{T_\card}}{\E{T_\pi}} \leq 1$ for any dispatching policy~$\pi$.
\end{restatable:theorem}

\begin{proof}
    See \cref{sec:upper} for the case of $n = 2$ servers and \cref{sec:proofs:multi} for the general case.
    \noqed
\end{proof}

\subsection{Intuition for Lower Bound on All Policies}

We now give some intuition for \cref{thm:rt_lower}. We focus on the heavy-traffic regime, where our aim is to show that the best possible mean response time is roughly $\E{T} \approx K_\card \E{W_\mgone}$.

To begin, recall $\E{T} = \frac{1}{\lambda} \E{N}$, where $\E{N}$ is the mean number of jobs in the system. The key idea is to relate $\E{N}$ to the mean amount of work $\E{W_\all}$. This is helpful because one can easily show $\E{W_\all} \geq \E{W_\mgone}$ (see e.g. \cref{thm:work_decomp}).

How can we relate $\E{N}$ to $\E{W_\all}$? In heavy traffic, most jobs in the system are waiting in a queue and have yet to enter service. We thus approximate $\E{N} \approx \E{W_\all} / \E{S_{\mathsf{queue}}}$, where $\E{S_{\mathsf{queue}}}$ is the mean size of jobs waiting in a queue. This means minimizing $\E{T}$ amounts to maximizing the mean size of jobs waiting in the queue. This makes sense in light of the fact that when studying scheduling policies beyond FCFS, serving small jobs ahead of large jobs reduces mean response time \citep{schrage1968proof}.

What is the largest that $\E{S_{\mathsf{queue}}}$ can be? Because we are restricted to FCFS service, the only mechanism by which we can affect the sizes of jobs in the system is dispatching. In particular, we can dispatch jobs of different sizes to different servers. Suppose, for example, that servers $1, \dots, n - 1$ have a negligible amount of work, meaning nearly all of the work is at server~$n$. Then $\E{S_{\mathsf{queue}}}$ would be the average size of jobs dispatched to server~$n$, which could be much greater than $\E{S}$. The best we could hope to do is $\E{S_{\mathsf{queue}}} = \E{S \given S \geq m}$ for as high a threshold~$m$ as possible. But in heavy traffic, we need server~$n$ to handle a $\frac{1}{n}$ fraction of the load, so the largest value of $m$ possible solves $\E{S \I(S \geq m)} = \frac{1}{n} \E{S}$. This is equivalent to the characterization of $m$ from~\cref{eq:Kcard}, so it leads to $\E{T} / \E{W_\all} \approx \E{S} / \E{S \given S \geq m} = K_\card$. Observing $\E{W_\all} \geq \E{W_\mgone}$ completes the bound.

To make this reasoning rigorous, it turns out that reasoning directly in terms of $\E{S_{\mathsf{queue}}}$ is difficult. We instead prove \cref{thm:rt_lower} using a potential-function approach. However, the potential function and manipulations we perform on it were directly inspired by the intuition:
\begin{quote}
    The best-case scenario is to dedicate one server to the jobs of size at least~$m$, and to ensure that all other servers have a negligible amount of work.
\end{quote}

\subsection{Intuition for Upper Bound on CARD}

We now give some intuition for \cref{thm:rt_upper_heavy}. By the lower bound intuition above, CARD is already well on its way to achieving the best-case scenario: it attempts to keep the amount of work at the $n - 1$ short servers near~$c$, and the long server only serves medium and large jobs. To show CARD matches the lower bound in heavy traffic, it would suffice to show the following.
\* CARD does not have much more work than a resource-pooled M/G/1: $\E{W_\all} \approx \E{W_\mgone}$.
\** Roughly speaking, this amounts to showing that we avoid situations where one server is idle while another server has lots of work (see \cref{thm:work_decomp}).
\* CARD's short servers do not exceed $c$ work by too much: $\E{W_s} \approx c$.
\** We also need to set $c$ such that it is negligible in heavy traffic.
\* CARD rarely dispatches medium jobs to the long server: $\P{W_s \leq c} \approx 1$.
\*/

Our main tool for showing these and related properties is examining what we call \emph{below-above cycles}. Consider a particular short server. It alternates between \emph{below periods}, during which $W_s \leq c$, and \emph{above periods}, during which $W_s > c$. It turns out that much of our analysis rests on below-above cycles not being too long. One reason for this is that when enough short servers are in above periods, the long server is temporarily overloaded. Long periods of transient overload could cause $\E{W_\all}$ to be significantly greater than $\E{W_\mgone}$. Short below-above cycles prevent this possibility. See \cref{sec:upper:ingredients} for more details about how we use below-above cycles.

\section{Universal Lower Bound}
\label{sec:lower}

\restate*\ref{thm:rt_lower}

Before diving into the proof, we give the high-level idea for $n = 2$ servers.

Suppose an arrival occurs while $W_1 < W_2$. For that individual arrival, its response time if it were sent to queue~$i$ would be~$2 W_i$ because each server processes work at rate $1/2$, so the ``benefit'' of sending it to queue~$1$ instead of queue~$2$ is $2(W_2 - W_1)$. Reasoning symmetrically if $W_1 < W_2$, we conclude that the benefit of dispatching jobs to the shorter queue is proportional to $|W_2 - W_1|$.

The main challenge is therefore to show that no dispatching policy can both frequently dispatch to the shorter queue, and also maintain large difference $|W_2 - W_1|$ between the queues. The key observation is that if we dispatch the job to the shorter queue, then $|W_2 - W_1|$ decreases, so the next arrival would see less benefit. That is, we can view $|W_2 - W_1|$ as a type of resource: dispatching jobs to the shorter queue depletes it, while dispatching jobs to the longer queue replenishes it. It is thus best to dispatch shorter jobs to the shorter queue, which slowly depletes $|W_2 - W_1|$, and dispatch longer jobs to the longer queue, which quickly replenishes $|W_2 - W_1|$. To formalize the idea of viewing $|W_2 - W_1|$ as a resource, we use the potential function $\tfrac{1}{2}(W_2 - W_1)^2$.

The proof below handles any number of servers~$n$. The idea is essentially the same as the $n = 2$ case, except we look at the work differences $|W_i - W_j|$ for every pair of servers~$i \neq j$.

\begin{proof}[Proof of \cref{thm:rt_lower}]
    Consider an arbitrary stationary dispatching policy~$\pi$. We first introduce notation for $\pi$'s dispatching decisions. Suppose a job of random size~$S$ arrives and observes work vector $\mathbf{W} = (W_1, \dots, W_n)$. We denote by $W_\choice$ the work at the queue the arrival is dispatched to. Note that while $S$ is independent of $\mathbf{W}$, it is \emph{not} independent of $W_\choice$. We also write $W_\all = \sum_{i = 1}^n W_i$ for the total work at all queues. Because each server does work at rate $1/n$, we can write $\E{T_\pi}$ as
    \[
        \E{T_\pi}
        = n \E{W_\choice + S}
        \label{eq:multi:rt_lower:delay}
        = \E{W_\all} + \E{n W_\choice - W_\all} + n \E{S}.
    \]

    The main task is to give a lower bound on $\E{n W_\choice - W_\all}$. To do so, we apply the rate conservation law of \citet{miyazawa1994rate} to $V(\mathbf{W})$, where
    \[
        V(\mathbf{w}) = \frac{1}{2} \sum_{i = 1}^n \sum_{j = 1}^{i - 1} (w_i - w_j)^2.
    \]
    The value of $V(\mathbf{W})$ can change in two ways.
    \* Work is done continuously at each nonempty queue. We denote this average continuous change by $\E{D_t V(\mathbf{W})}$.
    \* Arrivals add work to whichever queue the dispatcher chooses. By PASTA (Poisson Arrivals See Time Averages) \citep{wolff1982poisson}, this yields average change $\lambda \E{V(\mathbf{W} + S \mathbf{e}_\choice) - V(\mathbf{W})}$, where $\mathbf{e}_\choice$ is the standard basis vector with a $1$ indicating the queue the job is dispatched to.
    \*/
    The rate conservation law \citep{miyazawa1994rate} states that the average rate of change of $V(\mathbf{W})$ is zero, so
    \[
        \label{eq:multi:rt_lower:rcl}
        \E{D_t V(\mathbf{W})} + \lambda \E{V(\mathbf{W} + S \mathbf{e}_\choice) - V(\mathbf{W})} = 0.
    \]

    We now investigate each of the two terms in \cref{eq:multi:rt_lower:rcl}. We first observe that $\E{D_t V(\mathbf{W})} \leq 0$, because in the absence of arrivals, for any two queues~$i$ and~$j$, the absolute difference $|W_i - W_j|$ either decreases (if exactly one server is idle) or stays constant (otherwise). Therefore,
    \[
        \E{V(\mathbf{W} + S \mathbf{e}_\choice) - V(\mathbf{W})} \geq 0.
    \]
    Expanding the definition of $V(\mathbf{w})$ and writing $\sum_{i \neq \mathsf{choice}}$ for sums over all queues other than the one the job is dispatched to, we obtain
    \[
        0
        &\leq \frac{1}{2} \E*{\smashoperator[r]{\sum_{i \neq \mathsf{choice}}} \gp[\big]{(W_\choice + S - W_i)^2 - (W_\choice - W_i)^2}} \\
        &= \frac{n - 1}{2} \E{S^2} + \E*{\smashoperator[r]{\sum_{i \neq \mathsf{choice}}} S(W_\choice - W_i)} \\
        &= \frac{n - 1}{2} \E{S^2} + \E{S (n W_\choice - W_\all)}.
    \]
    Subtracting both sides from $m \E{n W_\choice - W_\all}$ and using the fact that
    \[
        -W_\all \leq n W_\choice - W_\all \leq (n - 1) W_\all,
    \]
    we obtain
    \[
        m \E{n W_\choice - W_\all}
        &\geq \E{(m - S) (n W_\choice - W_\all)} - \frac{n - 1}{2} \E{S^2} \\
        &\geq -\E[\big]{(S - m)^+ (n - 1) W_\all - (m - S)^+ W_\all} - \frac{n - 1}{2} \E{S^2} \\
        &\eqnote{(a)}{=} -\gp[\big]{\E[\big]{(n - 1) (S - m)^+ + (m - S)^+} \, \E{W_\all} + \frac{n - 1}{2} \E{S^2}} \\
        \label{eq:multi:rt_lower:benefit_bound}
        &= -\gp[\big]{(m - \E{S} + n\E{(S - m)^+}) \E{W_\all} + \frac{n - 1}{2} \E{S^2}},
    \]
    where (a) follows from the fact that an arriving job's size~$S$ is independent of the work vector~$\mathbf{W}$ it observes upon arrival.

    We now substitute the bound from \cref{eq:multi:rt_lower:benefit_bound} into \cref{eq:multi:rt_lower:delay}, obtaining
    \[
        \E{T_\pi}
        &= \frac{\E{S} - n \E{(S - m)^+}}{m} \E{W_\all} - \frac{(n - 1) \E{S^2}}{2 m} + n \E{S}.
    \]
    The bound follows from $\E{W_\all} \geq \E{W_\mgone}$ (see e.g. \cref{thm:work_decomp}) and \cref{eq:Kcard}, which implies
    \[
        \E{S} - n \E{(S - m)^+}
        = \E{S} - n \E{S \I(S > m)} + m n \P{S > m} 
        = m n \P{S > m} 
        = m K_\card.
        \qedhere
    \]
\end{proof}

\section{CARD Stability Analysis}
\label{sec:stability}
Proving CARD's stability is more than a straightforward application of the Foster-Lyapunov theorem, which is widely used to establish stability of queueing systems. The main obstacle here is that the long server alternates between being underloaded and overloaded. It is thus difficult to find a Lyapunov function that is negative outside a compact set.

To overcome this obstacle, we use a result of \citet[Theorem 1]{foss2012stability}. Notice that, under CARD, $W_s$ is itself a Markov process because the decision of where to dispatch a job depends only on the work at the shorter server. Roughly, \citep[Theorem 1]{foss2012stability} says that since $W_s$ is a Markov process of its own, if it is ergodic, then it suffices to do a drift analysis of $W_\ell$, \emph{averaged over} the stationary distribution of~$W_s$. Of course, we first need to show that $W_s$ is ergodic. Our proof for CARD's stability therefore proceeds in three steps.
\* We show that the short server's work~$W_s(t)$, as a Markov process of its own, is Harris ergodic (\cref{thm:stability_short}).
\* With the stability of $W_s(t)$ in hand, we bound the idleness probability of the short server in steady state (\cref{thm:ssc_short,thm:idleness_short}).
\* We apply the result of \citet[Theorem 1]{foss2012stability} (\cref{thm:stability}) to show stability whenever the long server is \emph{on average} not overloaded. Our bound on the short server's idleness probability from the previous step thus gives a sufficient condition for stability.
\*/
Armed with these key ideas, the proofs themselves are relatively straightforward, with the bulk of the work being computation. As such, we defer most of these computation details to \cref{sec:proofs:stability}.

\begin{restatable:lemma}
    \label{thm:stability_short}
    $W_s$ is Harris ergodic for any $\epsilon>0$.
\end{restatable:lemma}
\begin{proof}[Proof sketch]
The proof uses a Foster-Lyapunov theorem for continuous-time Markov processes \citep[Theorem 4.2]{meyn1993stability3}. The key step is to verify that the Lyapunov function $V(w_s) = w_s$ has bounded drift when $w_s \leq c$ and negative drift when $w_s > c$. This is true because when $w_s > c$, we only send small jobs to the short server. We defer the details to \cref{pf:stability_short}.
\noqed
\end{proof}

We establish our short server idleness bound by first proving a general bound on the probability that $W_s$ is lower than $c$ by a general amount~$x$. The idleness bound follows by plugging in $x = c$.

\begin{restatable:lemma}
    \label{thm:ssc_short}
    Suppose $\theta > 0$ satisfies $\widetilde{(S_{s,m})_e}(\theta)>\frac{1}{n(n-1)\beta+n-1}$, where $\widetilde{(S_{s,m})_e}(\cdot)$ is the Laplace transform of the equilibium distribution of the size of small and medium jobs, $S_{s,m} = (S \given S < m_+)$. Then for all $x \in [0, c]$,
    \[
        \P{W_s < c-x} \leq \frac{\gp*{n(n-1)\beta+n-1}\widetilde{(S_{s,m})_e}(\theta)}{\gp*{n(n-1)\beta+n-1}\widetilde{(S_{s,m})_e}(\theta)-1}e^{-\theta x}.
    \]
\end{restatable:lemma}
\begin{proof}[Proof sketch]
This result is a Chernoff-type bound on $(c - W_s)^+$, so the main task is to bound~$\E{\exp\gp{\theta (c - W_s)^+}}$. We do this by applying the rate conservation law \citep{miyazawa1994rate} to $\exp\gp{\theta (c - W_s)^+}$. We defer the details to \cref{pf:ssc_short}.
\noqed
\end{proof}

\begin{restatable:lemma}
    \label{thm:idleness_short}
    We have the following bound on the idleness of the short server,
    \[
        \P{W_s = 0} \leq \frac{n+1}{n\beta}\exp\gp*{-\frac{\beta c}{n(n-1)m_+}}.
    \]
\end{restatable:lemma}
\begin{proof}
Let $\theta=\frac{\beta}{n(n-1)}\frac{1}{m_+}$. Since $\beta < \tfrac{1}{n(n-1)}$ and all small and medium jobs have length at most~$m_+$, we have$\widetilde{(S_{s,m})_e}(\theta)\geq1-\theta\E{(S_{s,m})_e}\geq1-\frac{\beta}{n(n-1)}$. We can therefore apply \cref{thm:ssc_short}, from which the bound follows by the computation below and setting $x = c$:
\[
    \P{W_s<c-x}
    &\leq \frac{\gp*{n(n-1)\beta+n-1}\gp*{1-\frac{\beta}{n(n-1)}}}{\gp*{n(n-1)\beta+n-1}\gp*{1-\frac{\beta}{n(n-1)}}-1}\exp\gp*{-\frac{\beta x}{n(n-1)m_+}}\\
    &\leq \frac{n+1}{n\beta}\exp\gp*{-\frac{\beta x}{n(n-1)m_+}}.
    \qedhere
\]
\end{proof}

We defer the proof of \cref{thm:stability} to \cref{pf:stability}.

\section{CARD Mean Response Time Analysis}
\label{sec:upper}

With the lower bound from \cref{thm:rt_lower} in mind, our next step is to establish an upper bound on the mean response time under CARD. We focus here on the two-server case. The general case uses the same ideas but has more complicated computations, so we defer its proof to \cref{sec:proofs:multi}.

Let $\E{T_{\card, s}}$, $\E{T_{\card, m}}$, and $\E{T_{\card, \ell}}$ be the mean response times of small, medium, and large jobs under CARD, respectively. We have
\[
\E{T_{\card}}
&= p_s \E{T_{\card, s}} + p_m\E{T_{\card, m}} + p_\ell\E{T_{\card, \ell}} \\
\label{eq:rt_CARD}
&\leq 2 \E{S} + 2 p_s \E{W_s} + 2p_m c \P{W_s \leq c} + 2 p_m \E{W_\ell \I(W_s > c)} + 2 p_\ell \E{W_\ell}.
\]
where the inequality follows from how CARD dispatches jobs, the PASTA property \citep{wolff1982poisson}, and the fact that the servers complete work at rate $1/2$. The main difficulty of analyzing~\eqref{eq:rt_CARD} lies in bounding $\E{W_\ell}$ and $\E{W_\ell \I(W_s>c)}$. We now give a high-level overview of the obstacles and our approach.

\subsection{Key Ingredients: Work Decomposition, Below-Above Cycles, and Palm Inversion}
\label{sec:upper:ingredients}

To bound $\E{W_\ell}$, it suffices to bound $\E{W_\all}$. The following theorem, called the \emph{work decomposition law} \citep{scully2020gittins, scully2022new}, provides a way to bound~$\E{W_\all}$. We state it below in a way that is specialized to our system.

\begin{restatable:theorem}
    \label{thm:work_decomp}
    Denote by $I$ the fraction of servers that are idle in steady state, namely
    \[
    I = \frac{1}{n} \sum_{i=1}^n \I(W_i=0).
    \]
    The steady-state mean total work $\E{W_\all}$ satisfies
    \[
        \E{W_\all} = \E{W_\mgone} + \frac{\E{I W_\all}}{\epsilon} = \frac{\frac{\lambda}{2} \E{S^2} + \E{I W_\all}}{\epsilon},
    \]
    where $\E{W_\mgone}$ is the work in an M/G/1 with arrival rate $\lambda$ and job size distribution~$S$.
\end{restatable:theorem}

The key component we need to bound from \cref{thm:work_decomp} is $\E{I W_\all}$. We would like to study
\[
I W_\all&=(W_s+W_\ell)\gp*{\tfrac{1}{2}\I(W_s=0)+\tfrac{1}{2}\I(W_\ell=0)}\\
&=\tfrac{1}{2}W_\ell\I(W_s=0)+\tfrac{1}{2}W_s\I(W_\ell=0)
\]
The main difficulty here is to bound $\E{W_\ell\I(W_s=0)}$. Since CARD dispatches differently to the long server based on the state of the short server, $W_\ell$ depends on the state of $W_s$. Such a dependency also poses challenges in analyzing $\E{W_\ell\given W_s>c}$, when even knowing $\E{W_\ell}$ is not sufficient.

Under CARD, $W_s$ alternates between being above and below the threshold~$c$. Such a behavior naturally leads to renewal intervals consists of the ``above'' periods and ``below'' periods.

\begin{definition}\label{def:periods_below+above}
    We partition time into alternating intervals, called \emph{below periods} and \emph{above periods}, as follows:
    \* A time~$t$ is in a \emph{below period} if $W_s(t) \leq c$.
    \* A time~$t$ is in an \emph{above period} if $W_s(t) > c$.
    \*/
    A \emph{below-above cycle} is then a complete below period followed by a complete above period. Below-above cycles start at times~$t$ for which $W_s(t) = c$. We can partition time into below-above cycles.
    
    We introduce the following notation for working with below periods, above periods, and below-above cycles:
    \* We write $\E_c^0{\cdot}$ for the Palm expectation \citep{baccelli2002elements} taken at the start of a below-above cycle. Roughly speaking, $\E_c^0{\cdot} = \text{``}\E{\cdot \given \text{a below period starts at time~$0$}}\text{''}$, but the formal definition avoids conditioning on a measure-zero event.
    \* In the context of a below-above cycle starting at time~$0$, meaning $W_s(0) = c$, we denote the lengths of the below and above period by $B$ and~$A$, respectively:
    \[
        B &= \inf\curlgp{t > 0 : W_s(t) > c}, \\
        A &= \inf\curlgp{t > B : W_s(t) = c}.
    \]
    Abusing notation slightly, we also use $B$ and~$A$ to denote the lengths of the below and above period in a generic below-above cycle, not necessarily one that starts at time~$0$.
    \*/
\end{definition}

Why are above and below periods helpful for analyzing CARD? Within an above or below period, CARD does not change how it dispatches jobs, making it easier to analyze $W_\ell$ within \emph{one} below-above cycle. The Palm inversion formula \citep{baccelli2002elements}, which is a generalization of the celebrated renewal-reward theorem, allows us to connect the average behavior of $W_\ell$ within one below and above cycle to a steady-state average. For example, it implies
\[
    \E{W_\ell} = \frac{1}{\E{A+B}}\E_c^0*{\int_0^{A+B}W_\ell(t)\d{t}}, \qquad
    \E{W_\ell\I(W_s>c)} = \frac{1}{\E{A+B}}\E_c^0*{\int_B^{A+B}W_\ell(t)\d{t}}.
\]
Our high-level idea is to relate both of these quantities to $\E_c^0{W_\ell(0)}$, the mean work at the long server at the start of a below-above cycle. We show in \cref{thm:Wl_work_change, thm:work_long_below+above} that, roughly speaking,
\[
    \label{eq:work_comparison_intuition}
    \E{W_\ell} \approx \E_c^0{W_\ell(0)}, \qquad
    \E{W_\ell \I(W_s>c)} \approx \E_c^0{W_\ell(0)} \, \P{W_s > c}.
\]

The rest of this section is organized as follows.
\* \cref{sec:upper:short} analyzes the behavior of the short server. In particular, we show that above and below cycles are not too long.
\* \cref{sec:upper:long} analyzes the behavior of the long server. Using the fact that above and below cycles are not too long, we show \cref{eq:work_comparison_intuition}. As part of this, we bound~$\E{W_\all}$.
\* \cref{sec:upper:rt} assembles the pieces to prove \cref{thm:rt_upper_heavy}.
\*/

\subsection{Analyzing the Short Server and Below-Above Cycles}
\label{sec:upper:short}
In this section, we bound various quantities relating to work at the short server and the below-above cycles. Of particular importance are the mean \emph{excesses} of the above and below periods $\E{A_e}$ and $\E{B_e}$, as they are used to better understand the relations between $\E{W_\ell}$ and $\E_c^0{W_\ell(0)}$.

The techniques we use to obtain bounds on $\E{A_e}$ and $\E{B_e}$ also immediately yield bounds on $\E{A}$ and~$\E{B}$. Despite not using these bounds, given that they help complete the picture of how the system behaves, we state them, too.

As a reminder, the \emph{excess} or \emph{equilibrium distribution} of a random variable~$V$ is the distribution $V_e$ whose probability density function is $f(t) = \P{V > t} / \E{V}$. The excess arises naturally in renewal theory \citep{baccelli2002elements, harchol2013performance, ross1995stochastic}. Most important for our purposes is the fact that
\[
    \label{eq:mean_excess}
    \E{V_e} = \frac{\E{V^2}}{2 \E{V}}.
\]

\begin{restatable:lemma}
    \label{thm:mean_excess_bound_below}
    \[
        \E{B} \leq \frac{m_+}{\beta}, \qquad
        \E{B_e} \leq \frac{c+m_+}{\beta} \leq \frac{2 c}{\beta}, \qquad
        \text{and} \qquad
        \E{(B_e)_e} \leq \frac{c+m_+}{\beta} \leq \frac{2 c}{\beta}.
    \]
\end{restatable:lemma}

\begin{proof}
Suppose that at time~$0$, the short server has $W_s(0) = v \leq c$ work, so time~$0$ is in a below period. Let $\tau(v)$ be the time until the end of the below period. We will show
\[
\label{eq:hitting_time_below_process}
\E{\tau(v)} \leq \frac{c + m_+ - v}{\beta} \leq \frac{c + m_+}{\beta} \leq \frac{2 c}{\beta},
\]
where the last step follows because $c \geq m_+$ (\cref{sec:model:card}). This implies all three of the bounds.
\* A below period starts with $c$ work at the short server, so $\E{B} = \E{\tau(c)} \leq \frac{m_+}{\beta}$.
\* The excesses $B_e$ and~$(B_e)_e$ can both be interpreted as the distribution of the amount of time until the below period ends, starting from some random amount of work at the short server, so their means can each be written as $\E{\tau(V)}$ for an appropriate variable~$V$.
\*/

It remains only to show \cref{eq:hitting_time_below_process}, which we do using a supermartingale argument. Suppose $W_s(0) = v$ as above, and define $X(t) = c - W_s(t) + \beta t$. We now show that $X(t)$ is a supermartingale with respect to the Markov process~$W_s(t)$. Let
\* $\Delta_s(u, t)$ be the amount of work completed by the short server during $(u, t]$ and
\* $\Sigma_s(u, t)$ be the amount of work that arrives to the short server during $(u, t]$.
\*/
For any $0 \leq u \leq t$, we have
\[
\E{X(t) \mid W_s(u)} - X(u)
&= \E{W_s(u) - W_s(t) \mid W_s(u)} + \beta (t - u) \\
&= \E{\Delta_s(u, t) - \Sigma_s(u, t) \mid W_s(u)} + \beta (t - u) \\
&\leq \E*{\tfrac{1}{2}(t - u) - \Sigma_s(u, t) \mid W_s(u)} + \beta (t - u) \\
&= (t - u) \gp*{\tfrac{1}{2} - \rho_s - \rho_\ell + \beta} = 0,
\]
so $X(t)$ is indeed a supermartingale. Applying the optional stopping theorem to~$X(t)$ and $\tau(v)$, which we justify below, yields
\[
    c - v
    = \E{X(0)}
    \geq \E{X(\tau(v))}
    = c - W_s(\tau(v)) + \beta \tau(v)
    \eqnote{(a)}{\geq} -m_+ + \beta \tau(v),
\]
from which \cref{eq:hitting_time_below_process} follows. Above, (a) uses the fact that all medium jobs have size at most~$m_+$, so at the moment the below period ends, the short server's work can jump to at most $c + m_+$.

All that remains is to verify that we can indeed apply the optional stopping theorem.
\* We have $\E{\tau(v)} < \infty$ by positive recurrence of~$W_s(t)$.
\* We have uniform integrability, namely $\lim_{t\to\infty} \E{X(t) \, \I(\tau(v) > t)} = 0$, thanks to the following two observations. First, $\E{W_s(t) \, \I(\tau(v) > t)} \to 0$ because $c - W_s(t) \in [0, c]$ when $t$ is in a below period. Second, $\E{\beta t \, \I(\tau(v) > t)} \leq \E{\beta \tau(v) \, \I(\tau(v) > t)} \to 0$ because $\E{\tau(v)} < \infty$.
\qedhere
\*/
\end{proof}

\begin{restatable:lemma}
    \label{thm:work_short}
    \[
        \E{W_s - c \given W_s > c}
        \leq \frac{m_+}{4 \alpha}
        \qquad
        \text{and} \qquad
        \E{(W_s - c)^2 \given W_s > c}
        \leq \frac{m_+^2}{8 \alpha^2}
    \]
\end{restatable:lemma}

\begin{proof}[Proof sketch]
    Each above period starts with $W_s - c \in [0, m_+]$. Until the end of the above period, $W_s - c$ evolves like the amount of work in an M/G/1 queue with server speed~$1/2$, job size distribution~$S_s$, and work arrival rate $\rho_s < 1/2$. This means $(W_s - c \given W_s > c)$ has the same distribution as an M/G/1 with vacations, where the vacation length distribution is that of $W_s - c$ at the start of an above period. The desired bounds follow from the work decomposition formula for the M/G/1 with vacations \citep{fuhrmann1985stochastic} and the observation that both job sizes and vacation lengths are bounded by~$m_+$. We defer the details to \cref{pf:work_short}.
\end{proof}

\begin{restatable:lemma}
    \label{thm:mean_excess_bound_above}
    \[
        \E{A} \leq \frac{m_+}{\alpha} \qquad
        \text{and} \qquad
        \E{A_e} \leq \frac{m_+}{4 \alpha^2}.
    \]
\end{restatable:lemma}

\begin{proof}
    As in the proof sketch of \cref{thm:work_short}, we view the short server during an above period as an M/G/1 with server speed~$1/2$ and work arrival rate~$\rho_s$, so the mean drift of $W_s$ is $-(1/2 - \rho_s) = -\alpha$. By standard results for M/G/1 busy periods \citep{harchol2013performance}, starting from $W_s - c = v$, it takes $v/\alpha$ time in expectation for the above period to end.
    \* The $\E{A}$ bound follows from the fact that at the start of an above period, $W_s - c \leq m_+$, implying $\E{A} \leq m_+/\alpha$.
    \* The $\E{A_e}$ bound follows from the fact that the residual time of an above period is distributed as~$A_e$. But the residual time of an above period is the same as the amount of time until an above period ends starting from the stationary distribution of $W_s - c$ conditional on being in an above period. This means $\E{A_e} = \E{W_s - c \given W_s > c}/\alpha$, so the result follows from \cref{thm:work_short}.
    \qedhere
    \*/
\end{proof}

\subsection{Analyzing the Long Server}
\label{sec:upper:long}
In this section, we bound differences between $\E_c^0{W_\ell(0)}$ and $\E{W_\ell}$, $\E{W_\ell\I(W_s>c)}$, and $\E{W_\ell\I(W_s=c)}$, separately. These bounds will help us upper bound $\E{W_\all}$, thereby obtaining a bound on $\E{W_\ell}$.

Let $q_A$ and $q_B$ be the probabilities of being in an above or below period, respectively. That is,
\[
    \label{eq:prob_above_below_renewal}
    q_A = \P{W_s > c} = \frac{\E{A}}{\E{A + B}} \qquad
    \text{and} \qquad
    q_B = \P{W_s \leq c} = \frac{\E{B}}{\E{A + B}},
\]
where the expressions in terms of expectations of $A$ and $B$ follow from renewal-reward theorem.

\begin{restatable:lemma}
    \label{thm:Wl_work_change}
    \[
        \vgp[\big]{\E{W_\ell}-\E_c^0{W_\ell(0)}}
        \leq \gp*{\sqrt{q_A \E{A_e}} + \sqrt{q_B \E{B_e}}}^2
        \leq \frac{q_A m_+}{2 \alpha^2} + \frac{4 q_B c}{\beta}.
    \]
\end{restatable:lemma}
\begin{proof}
    The long server workload process can be described as
    \[
        \label{eq:Wl_process}
        W_\ell(t)=\gp*{W_\ell(0)-\Delta_\ell(0,t)+\Sigma_\ell^m(0,t)+\Sigma_\ell^\ell(0,t)}^+,
    \]
    where
    \* $\Delta_l(0,t)$ is the total work processed by the long server in $(0,t]$,
    \* $\Sigma_\ell^m(0,t)$ is the total work added to the long server from medium job arrivals in $(0,t]$, and
    \* $\Sigma_\ell^\ell(0,t)$ is the total work added to the long server due to large job arrivals in $(0,t]$.
    \*/
    Applying the Palm inversion formula \citep{baccelli2002elements} to $W_\ell$ gives
    \[
        \E{W_\ell}
        &= \frac{1}{\E{A+B}}\E_c^0*{\int_0^{A+B}\gp*{W_\ell(0)-\Delta_\ell(0,t)+\Sigma_\ell^m(0,t)+\Sigma_\ell^\ell(0,t)}^+\d{t}} \\
        &\eqnote{(a)}{=} \E_c^0{W_\ell(0)} + \frac{1}{\E{A+B}}\E_c^0*{\int_0^{A+B} \max\curlgp*{-\Delta_\ell(0,t)+\Sigma_\ell^m(0,t)+\Sigma_\ell^\ell(0,t), -W_\ell(0)}\d{t}},
    \]
    where (a) holds since $W_\ell(0)$, the amount of long server work at time~$0$, is independent of $A + B$, the length of the below-above cycle starting at time~$0$.
    
    We now bound $\E{W_\ell}-\E_c^0{W_\ell(0)}$ separately from above and below. To obtain a lower bound, we bound the integrand below by $-\Delta(t)$, obtaining
    \[
        \E{W_\ell} - \E_c^0{W_\ell(0)}
        \geq - \frac{1}{\E{A+B}} \E_c^0*{\int_0^{A + B}\Delta_\ell(0,t)\d{t}}
        \eqnote{(b)}{\geq} - \frac{1}{\E{A+B}} \E_c^0*{\int_0^{A + B} \frac{t}{2} \d{t}}
        = -\frac{\E{(A+B)^2}}{4 \E{A+B}},
    \]
    where (b) holds because the server completes work at rate~$\tfrac{1}{2}$ while it is busy. To obtain an upper bound, we bound the integrand above by $\Sigma^m_\ell(0, t) + \Sigma^\ell_\ell(0, t)$. We first bound its conditional expectation given $A$ and~$B$. Notice that $\Sigma^m_\ell(0, t) + \Sigma^\ell_\ell(0, t)$ consists of arrivals of large jobs during $(0, t]$ and medium jobs during $(B, t]$. Neither of these types of arrivals impacts the lengths of the above and below periods, so
    \[
        \label{eq:arrival_conditional_expectation}
        \E_c^0{\Sigma^m_\ell(0, t) + \Sigma^\ell_\ell(0, t) \given A, B} &= \rho_m (t - B)^+ + \rho_\ell t \leq t.
    \]
    From \cref{eq:arrival_conditional_expectation} and a computation similar to the lower bound, we obtain
    \[
        \E{W_\ell}-\E_c^0{W_\ell(0)} \leq \frac{\E{(A+B)^2}}{2 \E{A+B}}.
    \]
    Combining this with the lower bound, the result follows from \cref{eq:mean_excess}, \cref{eq:prob_above_below_renewal}, and  Cauchy-Schwarz:
    \[
        \frac{\E{(A+B)^2}}{2 \E{A+B}}
        \leq \frac{\E{A^2} + \sqrt{\E{A^2} \E{B^2}} + \E{B^2}}{2 \E{A+B}}
        = q_A \E{A_e} + 2 \sqrt{q_A \E{A_e} \, q_B \E{B_e}} + q_B \E{B_e}.
    \]
    To complete the proof, we use the AM-GM inequality (arithmetic mean $\geq$ geometric mean) on $\sqrt{q_A \E{A_e} \, q_B \E{B_e}}$, then apply our bounds on $\E{A_e}$ and $\E{B_e}$ from \cref{thm:mean_excess_bound_below, thm:mean_excess_bound_above}.
\end{proof}

\begin{restatable:lemma}
    \label{thm:work_long_below+above}
    \[
        \vgp[\big]{\E{W_\ell \I(W_s > c)} - q_A \E_c^0{W_\ell(0)}}
        \leq q_A \E{A_e} + 2 \sqrt{q_A \E{A_e} \, q_B \E{B_e}}
        \leq \frac{q_A m_+}{4 \alpha^2} + \frac{\sqrt{2 q_A q_B m_+ c}}{\alpha \sqrt{\beta}}.
    \]
\end{restatable:lemma}
\begin{proof}
    Similar to that of \cref{thm:Wl_work_change}. See \cref{pf:work_long_below+above}.
    \noqed
\end{proof}

\begin{restatable:lemma}
    \label{thm:IW_analysis}
    \[
    \E{W_\ell \I(W_s=0)}
    \leq \delta \E_c^0{W_\ell(0)} + \sqrt{\delta \E{B_e^2}}
    \leq \delta \E_c^0{W_\ell(0)} + \frac{2 c \sqrt{2 \delta}}{\beta}.
    \]
\end{restatable:lemma}
\begin{proof}
    Applying Palm inversion formula \citep{baccelli2002elements} to $W_\ell \I(W_s=0)$ yields
    \[
        \E{W_\ell\I(W_s=0)} = \frac{1}{\E{A+B}} \E_c^0*{\int_0^BW_\ell(t)\,\I(W_s(t)=0)\d{t}},
    \]
    where we can end the integral at $B$ because we only have $W_s(t) = 0$ during below periods, which corresponds to $t \in [0, B)$. We further expand the right-hand side using~\cref{eq:Wl_process}. No medium jobs are dispatched to the short server during below periods, so
    \[
    \MoveEqLeft
        \E{W_\ell\I(W_s=0)}
        \leq \frac{1}{\E{A+B}}\E_c^0*{\int_0^B(W_\ell(0)+\Sigma_\ell^\ell(0,t))\,\I(W_s(t)=0)\d{t}}\\
        &\eqnote{(a)}{=} \frac{\E_c^0{W_\ell(0)}}{\E{A+B}} \E_c^0*{\int_0^B\I(W_s(t)=0)\d{t}} + \frac{1}{\E{A + B}}\E_c^0*{\int_0^B\Sigma_\ell^\ell(0,t)\,\I(W_s(t)=0)\d{t}}.
    \]
    where (a) follows from the independence of $W_\ell(0)$ and $\int_0^B\I(W_s(t)=0)\d{t}$. To analyze the first term, we observe that by the Palm inversion formula \citep{baccelli2002elements} and \cref{thm:stability},
    \[\frac{1}{\E{A+B}}\E_c^0*{\int_0^B\I(W_s(t)=0)\d{t}} = \E{\I(W_s = 0)} = \P{W_s=0} \leq \delta.\]
    To analyze the second term, we apply \cref{eq:arrival_conditional_expectation}, yielding
    \[
    \E_c^0*{\int_0^B\Sigma_\ell^\ell(0,t)\,\I(W_s(t)=0)\d{t}}
    \leq \E_c^0*{\int_0^B t\I(W_s(t)=0)\d{t}}.
    \]
    The right-hand side is difficult to compute directly due to the dependency of $B$ and $W_s$. To resolve this, we apply the Palm inversion formula \citep{baccelli2002elements} to $B_a\I(W_s=0)$, where $B_a(t)$ is the age process of the below-above cycle, namely the amount of time since the current cycle began. This yields
    \[\frac{1}{\E{A+B}}\E_c^0*{\int_0^B t\I(W_s(t)=0)\d{t}} = \E{B_a\I(W_s=0)}\]
    Thus, to bound $\mathcal{T}_3$, it suffices to bound $\E{B_a\I(W_s=0)}$. By Cauchy-Schwarz,
    \[
    \E{B_a\I(W_s=0)}\leq\sqrt{\E{B_a^2}\P{W_s=0}}
    \eqnote{(b)}{=} \sqrt{\E{B_e^2}\P{W_s=0}}
    \eqnote{(c)}{\leq} \sqrt{\delta\E{B_e^2}},
    \]
    where (b) follows because $B_a$ has distribution~$B_e$, and (c) follows from \cref{thm:stability}. The result then follows from bounding $\E{B_e^2}$ using \cref{eq:mean_excess, thm:mean_excess_bound_below}.
\end{proof}

\begin{restatable:lemma}
    \label{thm:bound_on_W}
    \[
        \E{W_\ell}
        \leq \E{W_\all}
        \leq \gp[\bigg]{1 + \frac{\delta}{\epsilon}} \E{W_\mgone}
            + 2 c + \frac{m_+ \sqrt{q_A}}{2 \alpha \sqrt{\epsilon}}
            + \frac{4 c \sqrt{\delta}}{\alpha^2 \beta \epsilon}.
    \]
\end{restatable:lemma}
\begin{proof}
We use \cref{thm:work_decomp} to bound $\E{W_\all}$, which amounts to analyzing $\E{I W_\all}$. We have
\[
    I W_\all
    = (W_s+W_\ell) \gp*{\tfrac{1}{2}\I(W_s=0)+\tfrac{1}{2}\I(W_\ell=0)}
    = \tfrac{1}{2}W_\ell\I(W_s=0)+\tfrac{1}{2}W_s\I(W_\ell=0).
\]
Combining \cref{thm:Wl_work_change,thm:IW_analysis} and noting $\E{W_\ell} \leq \E{W_\all}$ yields a bound on $\E{W_\ell\I(W_s=0)}$:
\[
    \E{W_\ell\I(W_s=0)}
    \leq \delta \gp*{\E{W_\all} + \frac{q_A m_+}{4 \alpha^2} + \frac{4 q_B c}{\beta}} + \frac{2 c \sqrt{2 \delta}}{\beta}.
\]
To bound $\E{W_s\I(W_\ell=0)}$, we compute
\[
    \E{W_s\I(W_\ell=0)}
    &\eqnote{(a)}{\leq} \E{(c+(W_s-c)^+)\,\I(W_\ell=0)}\\
    &\eqnote{(b)}{\leq} c\P{W_\ell=0}+\sqrt{\E*{((W_s-c)^+)^2} \, \P{W_\ell=0}}\\
    &= c\P{W_\ell=0}+\sqrt{q_A\E*{(W_s-c)^2 \given W_s > c} \, \P{W_\ell=0}}\\
    &\eqnote{(c)}{\leq} 2\epsilon c+\frac{m_+ \sqrt{q_A \epsilon}}{2\alpha},
\]
where (a) follows from $W_s\leq c+(W_s-c)^+$, (b) follows from Cauchy-Schwarz, and (c) follows from \cref{thm:work_short} and the fact that $\epsilon = \tfrac{1}{2}\P{W_s=0} + \tfrac{1}{2}\P{W_\ell=0} \geq \tfrac{1}{2}\P{W_\ell=0}$. Combining the bounds on $\E{W_\ell\I(W_s=0)}$ and $\E{W_s\I(W_\ell=0)}$ with \cref{thm:work_decomp}, we obtain
\[
    \E{W_\all}
    &= \E{W_\mgone} + \frac{1}{\epsilon} \gp*{\tfrac{1}{2} \E{W_\ell\I(W_s=0)} + \tfrac{1}{2} \E{W_s\I(W_\ell=0)}} \\
    &\leq \E{W_\mgone}
        + c
        + \frac{m_+ \sqrt{q_A}}{4 \alpha \sqrt{\epsilon}}
        + \frac{\delta}{2 \epsilon} \E{W_\all} + \frac{q_A m_+ \delta}{8 \alpha^2 \epsilon} + \frac{2 q_B c \delta}{\beta \epsilon}
        + \frac{c \sqrt{2 \delta}}{\beta \epsilon}.
\]
The result follows after rearranging and simplifying. We use the fact that we have defined the parameters such that $c \geq m_+$ and $\delta \leq \epsilon$ (\cref{sec:model:card, sec:model:quantities}), which means $1/\gp[\big]{1 - \frac{\delta}{2 \epsilon}} \leq 1 + \frac{\delta}{\epsilon} \leq 2$. And, using the fact that $\alpha, \beta \leq \tfrac{1}{2}$, we loosely bound the terms with a $\sqrt{\delta}$ factor by
\[
    \frac{q_A m_+ \delta}{8 \alpha^2 \epsilon} + \frac{2 q_B c \delta}{\beta \epsilon} + \frac{c \sqrt{2 \delta}}{\beta \epsilon}
    \leq (q_A + q_B) \frac{c \delta}{2 \alpha^2 \beta \epsilon} + \frac{c \sqrt{2 \delta}}{\beta \epsilon}
    \leq \frac{2 c \sqrt{\delta}}{\alpha^2 \beta \epsilon}.
    \qedhere
\]
\end{proof}

\subsection{Bounding Mean Response Time}
\label{sec:upper:rt}

We now prove \cref{thm:rt_upper_heavy}, our main upper bound result. It follows as a corollary of a more explicit bound, which we state in \cref{thm:rt_upper_explicit}. To simplify the computations, we assume that $\beta \geq 2\delta$, but we could remove this assumption at the cost of slightly complicating the expressions.

\begin{restatable:lemma}
    \label{thm:prob_above_below}
    If $\beta \geq 2 \delta$, then $q_A \leq \frac{2 \beta}{\alpha + \beta}$ and $q_B \leq \frac{\alpha}{\alpha + \beta}$.
\end{restatable:lemma}
\begin{proof}
    The short server is stable, so the load of jobs arriving to it equals the average rate it completes work. This means $\rho_s + \rho_m \P{W_s \leq c} = \tfrac{1}{2} \P{W_s > 0}$. \cref{thm:stability} implies $\P{W_s > 0} \in [1 - \delta, 1]$, so the bound follows from the definitions of $\alpha$ and $\beta$ and the $\beta \geq 2 \delta$ assumption.
\end{proof}

\begin{restatable:theorem}
    \label{thm:rt_upper_explicit}
    In a system with $n = 2$ servers, if $\delta \leq \epsilon < \frac{1}{2}$ and $\beta \geq 2 \delta$, then by setting $c$ according to \cref{thm:stability}, CARD achieves mean response time bounded by
    \[
        \E{T_\card}
        &\leq \gp*{K_\card + \frac{4 \beta}{\alpha + \beta}} \gp*{1 + \frac{\delta}{\epsilon}} \E{W_\mgone} + 2 \E{S}
            \\* &\quad
            + 44 m_+ \max\curlgp[\Bigg]{
                \frac{\beta}{\alpha^2 (\alpha + \beta)}, 
                \sqrt{\frac{\beta}{\alpha^2 \epsilon (\alpha + \beta)}},
                \frac{\log\frac{3}{2 \beta \delta}}{\beta (\alpha + \beta)},
                \sqrt{\frac{\log\frac{3}{2 \beta \delta}}{\beta}},
                \frac{\sqrt{\delta} \log\frac{3}{2 \beta \delta}}{\alpha^2 \beta^2 \epsilon}
            }.
    \]
\end{restatable:theorem}
\begin{proof}[Proof sketch]
    It suffices to bound the work quantities on the right-hand side of \cref{eq:rt_CARD}.
    \* \Cref{thm:work_short} implies $\E{W_s} \leq c + q_A \E{W_s - c \given W_s > c} \leq c + \frac{q_A m_+}{\alpha}$.
    \* \Cref{thm:Wl_work_change, thm:work_long_below+above} imply, after some simplification,
    \[
        \E{W_\ell \I(W_s > c)}
        \leq q_A \E{W_\ell} + \frac{q_A m_+}{\alpha^2} + \frac{4 q_A q_B c}{\beta} + \frac{\sqrt{2 q_A q_B m_+ c}}{\alpha \sqrt{\beta}}.
    \]
    \*/
    We use these with \cref{thm:bound_on_W, thm:prob_above_below} to express the right-hand side in terms of $\alpha$, $\beta$, $\delta$, and $m_+$, then simplify. We defer the details to \cref{pf:rt_upper_explicit}.
\end{proof}

\begin{proof}[Proof of \cref{thm:rt_upper_heavy} for $n = 2$ servers]
    The bound follows directly from plugging the parameter choices into \cref{thm:rt_upper_explicit}, and comparing with the lower bound in \cref{thm:rt_lower} implies heavy-traffic optimality. But the main question is why these are the right ways to set the parameters.
    
    If we set $\delta = \Theta(\epsilon^d)$ for fixed~$d$, the only expression in \cref{thm:rt_upper_explicit} that is increasing as a function of $d$ is $\log\frac{3}{2 \beta \delta} = d \Theta\gp*{\log \frac{1}{\epsilon}}$. We thus ignore factors of $\sqrt{\delta}$ when determining $\alpha$ and~$\beta$. One can check at the end that $d \geq 3$ suffices.

    Observe that we want $\beta/\alpha \downarrow 0$ to ensure the multiplier of $\E{W_\mgone}$ approaches $K_\card$. If we substitute $\beta = \kappa \alpha$ into \cref{thm:rt_upper_explicit}, then for any fixed~$\kappa$, the resulting expression is a decreasing function of $\alpha$, so we set $\alpha = \Theta(1)$. With this choice, the largest terms from the maximum in \cref{thm:rt_upper_explicit} are $\Theta\gp[\big]{\sqrt{\beta / \epsilon}}$ and $\Theta\gp[\big]{\frac{1}{\beta}\log\frac{1}{\epsilon}}$, which are balanced by $\beta = \Theta\gp[\big]{\epsilon^{1/3} \gp[\big]{\log \frac{1}{\epsilon}}^{2/3}}$.
\end{proof}

\section{Simulations}
\label{sec:simulation}
We have established the optimality of CARD as load approaches capacity. In this section, we investigate the performance of CARD in moderate traffic via simulations. We aim to provide insights into the following questions with our simulations.
\* How good is CARD's performance compared with other dispatching policies in the literature?
\* Are there simple modifications of CARD that exhibit better performance in practice? 
\* CARD has three tunable parameters: $c$, $\alpha$, and $\beta$. The recipe provided in \cref{thm:rt_upper_heavy} is optimal in heavy traffic, but are there rules of thumb that work well beyond heavy traffic? How sensitive is CARD's performance to these parameters? 
\*/
In all of our simulations, we consider three benchmark policies: LWL, SITA-E\footnote{%
    SITA-E is the version of SITA that splits the load equaly among all the servers. One can improve its performance very slightly by using an unbalanced load split, a policy known as SITA-O. But computing the optimal split is generally challenging \citep{harchol2010balance}, and in our initial experiments, SITA-O did not significantly improve upon SITA-E.}
\citep{harchol1999choosing}, and Dice \citep{hyytia2022sequential}. Roughly, Dice lets the server with least work pick small jobs from the arrival stream, leaving the large jobs for servers with more work. We refer interested readers to \Cref{sec:dice} for details.\footnote{%
    The version of Dice we simulate differs slightly from the original version in \citep{hyytia2022sequential}, where Dice have thresholds that do not vary with load. We notice that constant thresholds lead to suboptimal performance for either low or high loads. Therefore, we incorporate load-dependent thresholds for Dice that lead to good performance across all loads simulated. With two servers, we use a threshold of the form $\eta\epsilon^{-1/3}$ for Dice, picking $\eta = 1.8, 5.2, 20$ for $\mathsf{cv} = 1, 10, 100$, respectively. With ten servers, we use thresholds $2m_i\epsilon^{-1/3}$.}
Of course, there are many more dispatching policies. We pick LWL and SITA-E because they are extensively studied, and we pick Dice because among all heuristics for size- and state-aware dispatching, it has the best empirical performance at high load \citep{hyytia2022sequential}.

Our simulations include job size distributions with exponential and heavier tails. Heavy-tail distributions are common in computer systems and networks (e.g. \citep{mahanti2013tale}) and the high mean response times they incur make a good dispatching policy essential. Throughout this section, we consider three Weibull distributions with mean 1 and coefficients of variation ($\mathsf{cv}$) 1, 10, and~100. We simulate 40 trials for each data point, with $10^7$ arrivals per trial for $\mathsf{cv} = 1$ and $\mathsf{cv} = 10$, and $3\times10^7$ arrivals per trial for $\mathsf{cv} = 100$. We show 95\% confidence intervals when wider than the marker size.

\subsection{Performance of CARD with Two Servers}
Although CARD as introduced in \cref{sec:model:card} is heavy-traffic optimal, we can improve its performance under moderate traffic with one small modification: instead of statically deciding which server is short and which is long, dynamically treat whichever server has less work as the short server. We call this variant \emph{Flexible CARD}, and call the original version \emph{Rigid CARD} to disambiguate.

\begin{figure}
\begin{subfigure}[b]{0.31\textwidth}
    \includegraphics[width=\linewidth]{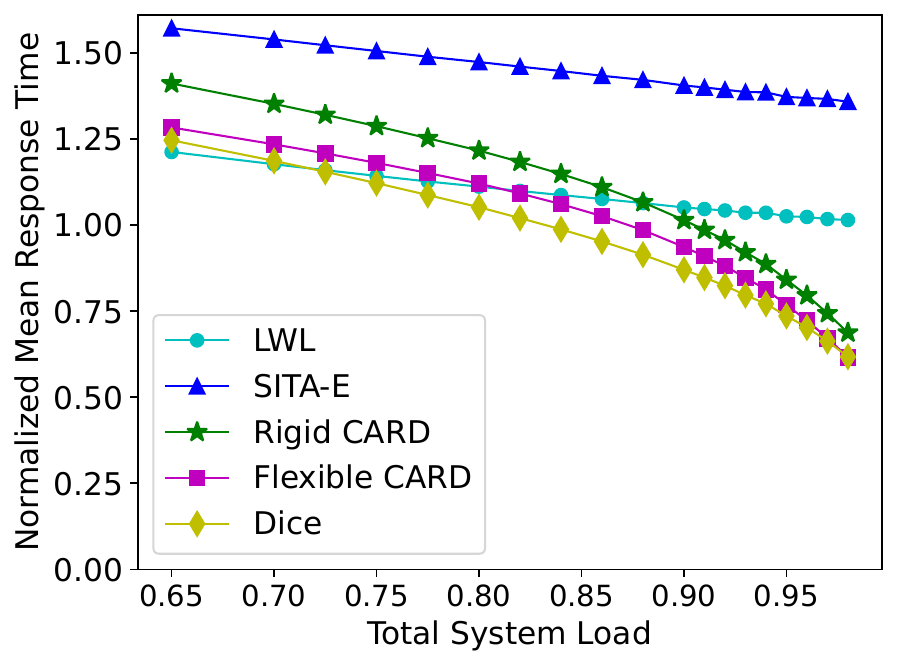}
    \caption{$2$ servers, $\mathsf{cv} = 1$}
\end{subfigure}\hfill
\begin{subfigure}[b]{0.31\textwidth}
    \includegraphics[width=\linewidth]{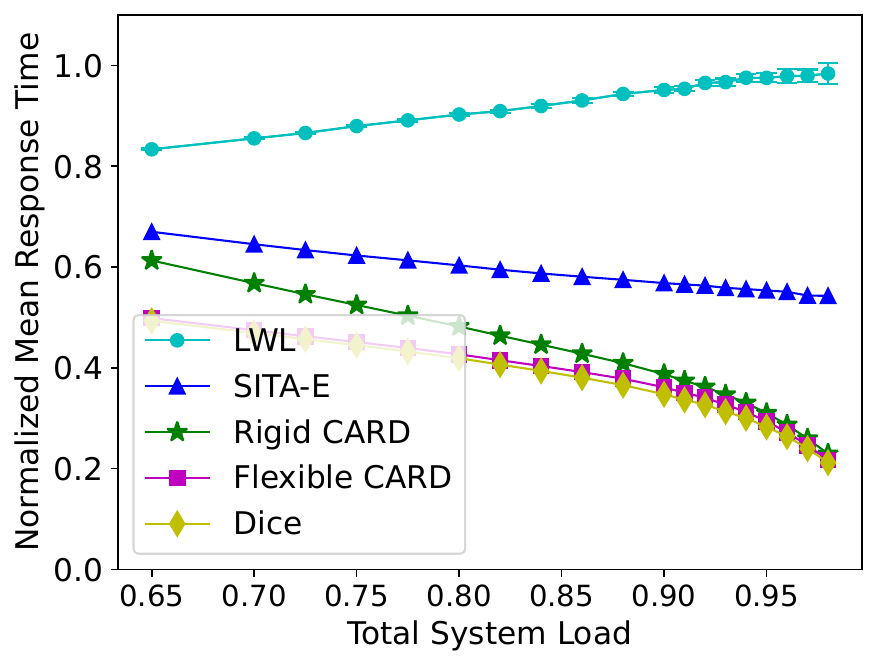}
    \caption{$2$ servers, $\mathsf{cv} = 10$}
\end{subfigure}\hfill
\begin{subfigure}[b]{0.31\textwidth}
    \includegraphics[width=\linewidth]{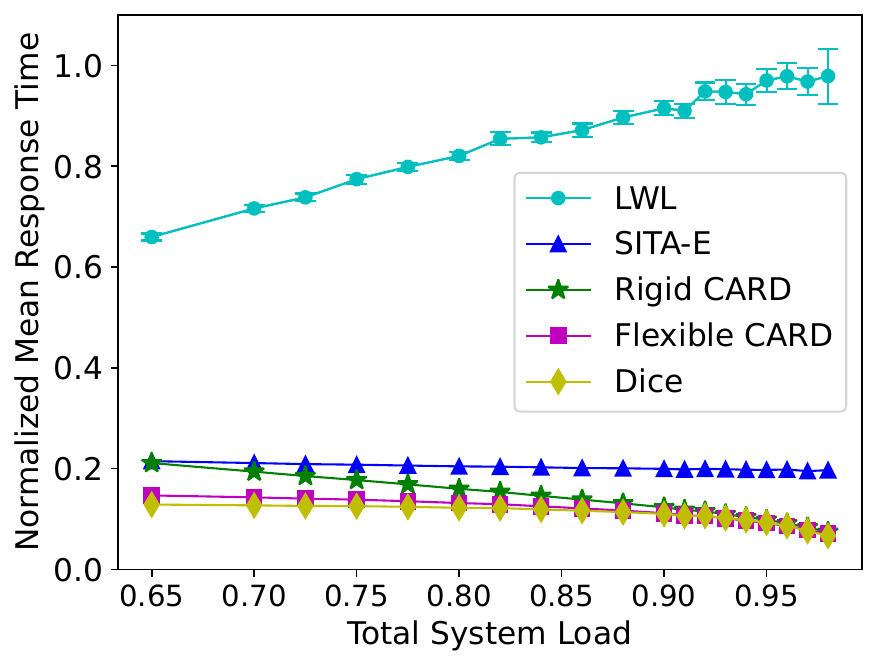}
    \caption{$2$ servers, $\mathsf{cv} = 100$}
\end{subfigure}
\caption{Normalized (relative to $\E{W_\mgone}$) mean response times for $n = 2$ servers.}
\label{fig:two-server-policies-comparison}
\end{figure}

\cref{fig:two-server-policies-comparison} shows us that both CARD versions significantly outperform LWL and SITA-E, especially at high loads and with large coefficients of variation.
For instance, with $\mathsf{cv} = 100$ and $\rho = 0.98$, CARD gives a \emph{93\% reduction} compared to LWL, and a \emph{61\% reduction} compared to SITA-E.
Flexible CARD is also almost tied with Dice at all loads simulated.

\subsection{Calibrating the Parameters of Two-Server CARD}
\label{sec:simulation:CARD_parameter}
We now discuss how to calibrate parameters $c$, $\alpha$, and $\beta$.
In practice, $\alpha$ and $\beta$ as prescribed in \cref{thm:rt_upper_heavy} are difficult to calibrate, because the ranges of $\alpha$ and $\beta$ change as $\rho$ increases.
Therefore, we consider instead the parameters $\alpha'=\frac{1}{2}-\frac{\rho_s}{\rho}$ and $\beta'=\frac{1}{2}-\frac{\rho_\ell}{\rho}$.
Adjusting $\alpha'$ can therefore be understood as adjusting the fraction of small jobs and adjusting $\beta'$ can be understood as adjusting the fraction of large jobs.

After trying a few strategies for scaling $c$ as a function of $\rho$, we found that thresholds of the form
$c=\gamma\frac{1}{\sqrt{\epsilon}}\log\gp*{\frac{1}{\epsilon}}$,
where $\gamma$ depends on the distribution, yield decent performance.

In general, for the three job size distributions we consider, mean response time under flexible CARD is not very sensitive to these parameters near the optima (see \cref{fig:flexible_CARD_parameter_sensitivity}). Any choice of parameters not too far from the optima yields decent performance. We found that $\alpha'=\beta'=0.15$ for all three distributions and $\gamma=0.3, 0.6, 2.5$ for cv = 1, 10, and 100, respectively, lead to decent performance. These are also the parameters we used in \cref{fig:two-server-policies-comparison}.

\begin{figure}
\begin{subfigure}[b]{0.31\textwidth}
    \includegraphics[width=\linewidth]{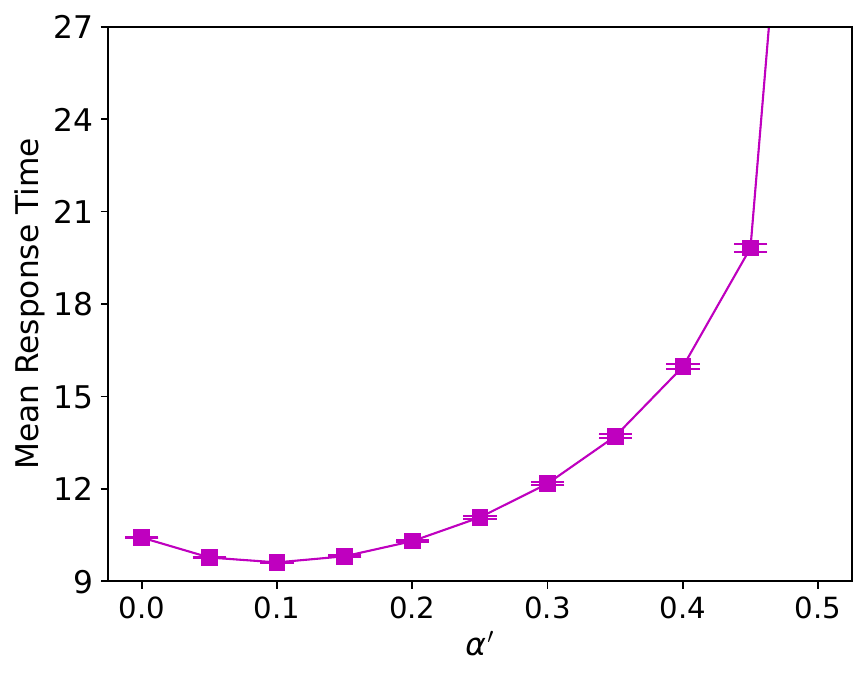}
    \caption{Varying $\alpha'$ ($\beta'=0.15$, $\gamma=0.6$)}
\end{subfigure}\hfill
\begin{subfigure}[b]{0.31\textwidth}
    \includegraphics[width=\linewidth]{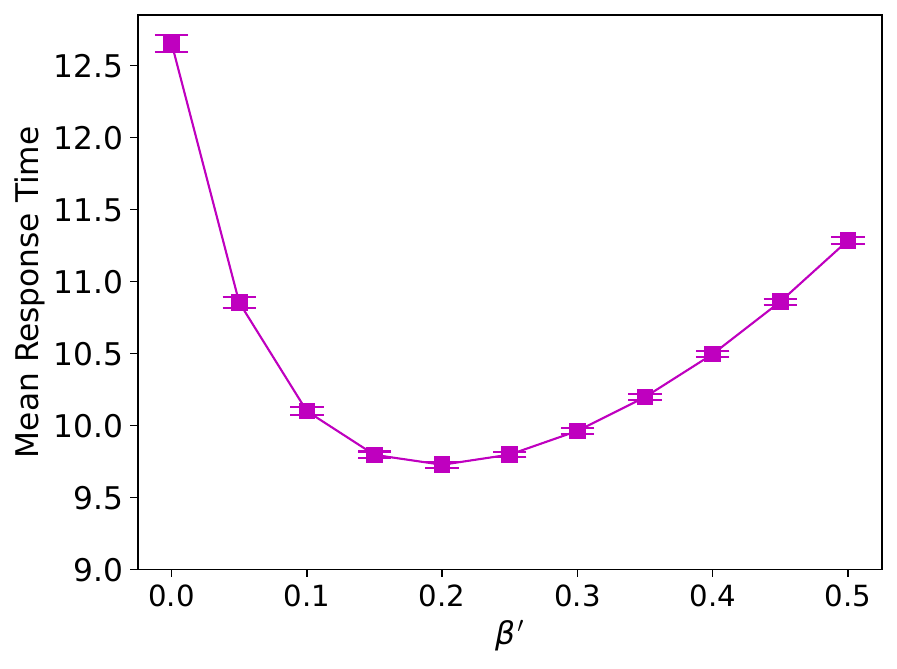}
    \caption{Varying $\beta'$ ($\alpha'=0.15$, $\gamma=0.6$)}
\end{subfigure}\hfill
\begin{subfigure}[b]{0.31\textwidth}
    \includegraphics[width=\linewidth]{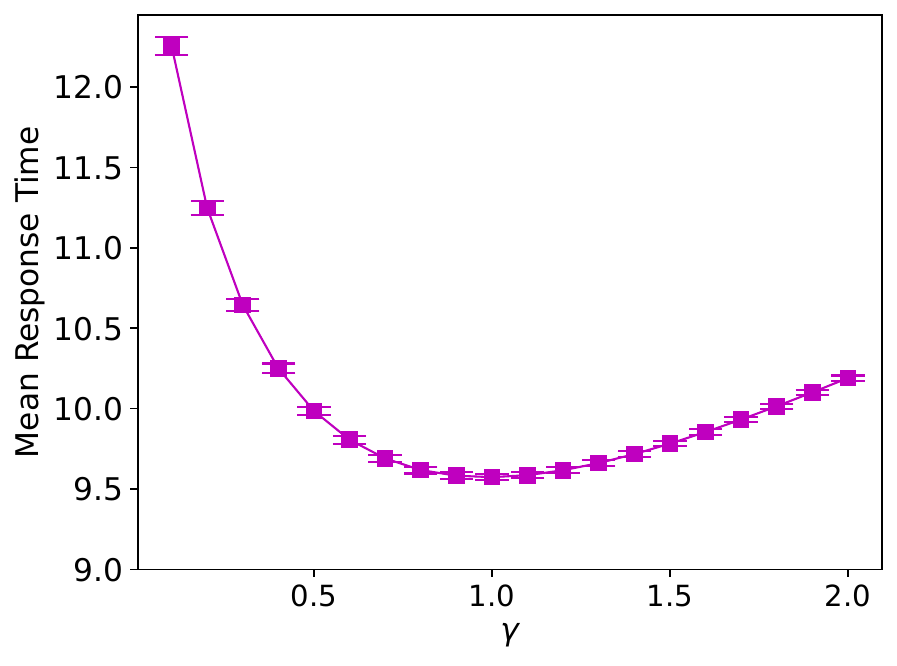}
    \caption{Varying $\gamma$ ($\alpha'=0.15$, $\beta'=0.15$)}
\end{subfigure}
\caption{For each of the three plots, we fix two parameters and vary one parameter across a range of values. Size distribution simulated has $\mathsf{cv} = 10$, and load is fixed at $\rho = 0.8$.}
\label{fig:flexible_CARD_parameter_sensitivity}
\end{figure}

\subsection{Improving CARD's Performance for More than Two Servers}
As the number of servers increases, flexible CARD with three parameters ($\gamma$, $\alpha'$, and $\beta'$) no longer performs well for distributions with large coefficients of variation \cref{fig:n-server-flexible}. Therefore, we propose another variant of CARD for $n$ servers called multi-band CARD. We first present the general dispatching rules, then explain how multi-band rigid and flexible CARDs are defined.
\* We divide the job size into $n+1$ small intervals such that each interval amounts to $\frac{1}{n}$ of the total load except for the first and last interval, each of which amounts to $\frac{1}{2n}$ of the total load. Denote the endpoints of these intervals as $0,m_1,\ldots m_n,\infty$.
\* Server $i$ except the last one has a threshold $c_i$, which can be different for different servers. When a job of size $s$ arrives, it is dispatched according to the following general rules:
    \** If $s<m_1$, it is dispatched to server 1.
    \** If $s>m_n$, it is dispatched to server $n$.
    \** If $s\in[m_i,m_{i+1})$ for $i=1,\ldots, n-1$, it is dispatched to server $i$ if $W_i\leq c_i$. Otherwise, it is dispatched to server $i+1$.
    \*/
\*/
Multi-band rigid CARD numbers the servers $1$ to $n$ and dispatches according to the rules outlined above. Server numbers do not change under rigid CARD. On the other hand, Multi-band flexible CARD sorts the servers in increasing work order when a job arrives so that $W_1\leq W_2\leq\cdots\leq W_n$, then dispatch according to the general rules.

Since all the $m_i$'s are fixed for each distribution, the tunable parameters are the $c_i$'s. Our experiments show that we achieve good performance by setting
$c_i = m_i/\sqrt{\epsilon}$.

As we can see in \cref{fig:n-server-policies-comparison},
multi-band CARDs significantly outperforms LWL and SITA-E at high loads and for job size with large coefficients of variation. When the job size distribution has cv=10, at $\rho=0.98$, mean response time under flexible CARD is $\sim$22\% and $\sim$19\% of the mean response times under LWL and SITA-E, respectively. When the job size distribution has cv=100, at $\rho=0.98$, mean response time under flexible CARD is $\sim$4\% and $\sim$21\% of the mean response times under LWL and SITA-E, respectively. Moreover, multi-band flexible CARD almost ties with Dice in all loads simulated.

\begin{figure}
\begin{subfigure}[b]{0.31\textwidth}
    \includegraphics[width=\linewidth]{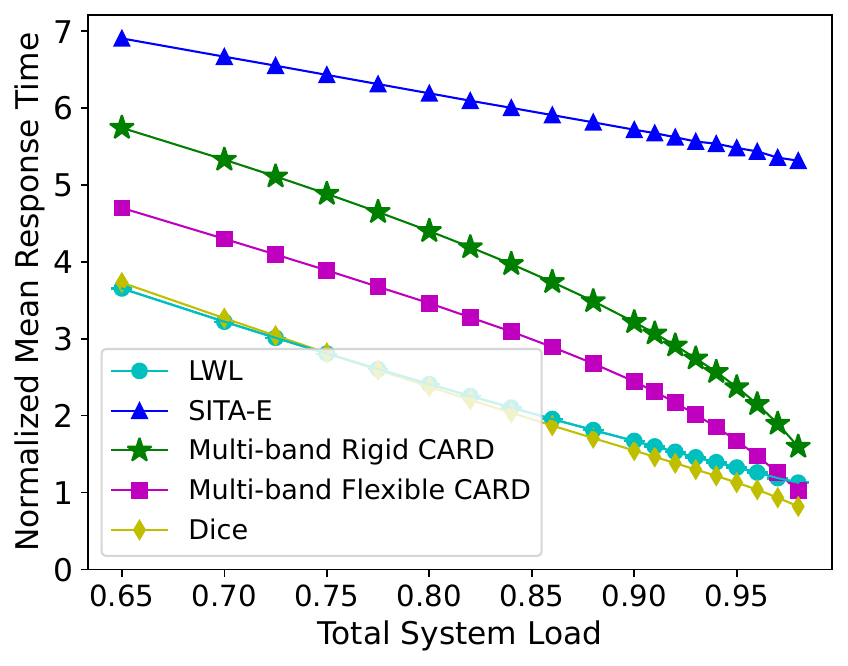}
    \caption{$10$ servers, $\mathsf{cv} = 1$}
\end{subfigure}\hfill
\begin{subfigure}[b]{0.31\textwidth} 
    \includegraphics[width=\linewidth]{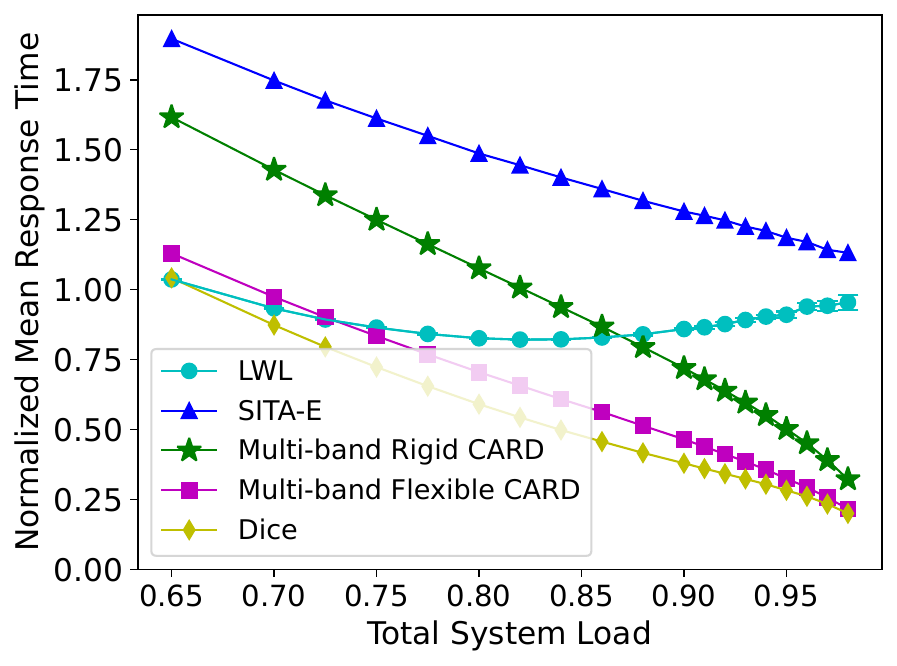}
    \caption{$10$ servers, $\mathsf{cv} = 10$}
\end{subfigure}\hfill
\begin{subfigure}[b]{0.31\textwidth}
    \includegraphics[width=\linewidth]{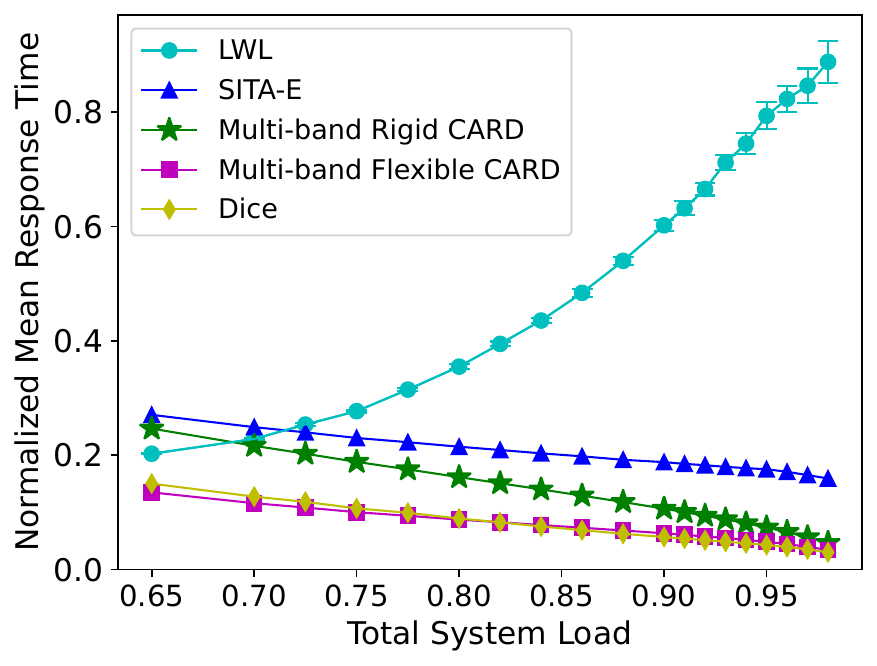}
    \caption{$10$ servers, $\mathsf{cv} = 100$}
\end{subfigure}
\caption{Normalized (relative to $\E{W_\mgone}$) mean response times for $n = 10$ servers.}
\label{fig:n-server-policies-comparison}
\end{figure}

\subsection{Tail Simulations}

Although our paper focuses exclusively on mean response time analysis, metrics based on the \emph{tail} of response time are often of interest in practice. As such, in this section, we conduct some simulations comparing the response time tails of CARD against the benchmark policies and provide some insights into the results.

We focus on a two-server system. The parameters of rigid CARD, flexible CARD, and Dice are the same as those in \Cref{sec:simulation:CARD_parameter}. The results are presented in \Cref{fig:two-server-tail-comparison}. We can see that for light-tail job size distribution, the tails of LWL and M/G/1/FCFS are better than those of rigid and flexible CARDs and Dice. On the other hand, for heavy-tail job size distribution, the tails of rigid and flexible CARDs and Dice are far better than those of LWL and M/G/1/FCFS up to 99-percentile of flexible CARD response time.

This result is not surprising. CARD and Dice both starve large jobs by making them wait in a long queue. For light-tail job size distributions, although giving a little priority to small jobs improves tail performance \citep{grosof2021nudge}, starving large jobs in general only worsens tail performance \citep{wierman2012tail}. For heavy-tail job size distributions, however, starving large jobs improves tail performance tremendously \citep{wierman2012tail}. As is shown in \Cref{fig:two-server-tail-comparison}, CARD and Dice significantly outperform LWL and M/G/1/FCFS up to 99-percentile of flexible CARD response time for heavy-tail job size distributions.

\begin{figure}
\begin{subfigure}[b]{0.31\textwidth}
    \includegraphics[width=\linewidth]{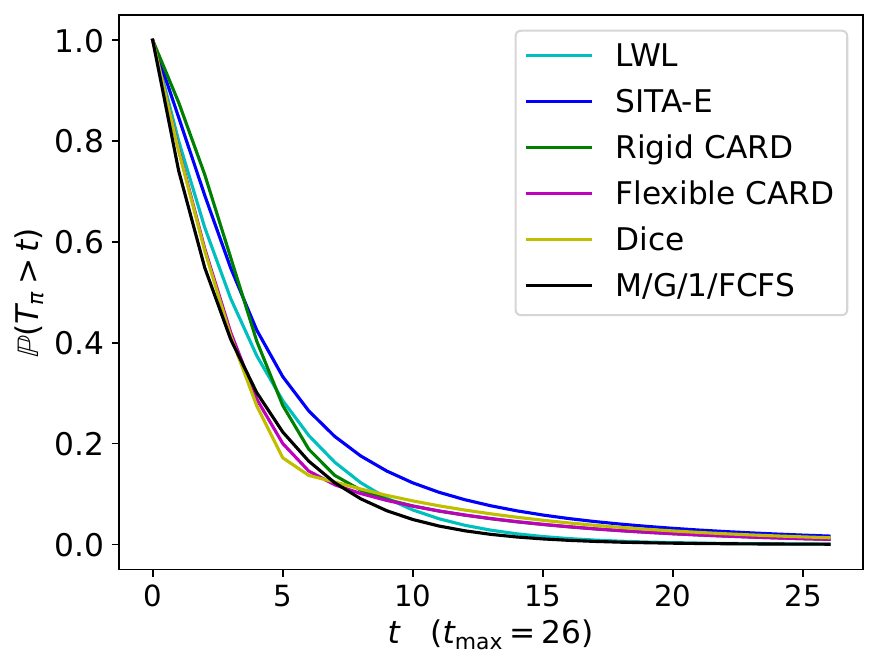}
    \caption{$2$ servers, $\mathsf{cv} = 1$, $\rho=0.7$}
\end{subfigure}\hfill
\begin{subfigure}[b]{0.31\textwidth}
    \includegraphics[width=\linewidth]{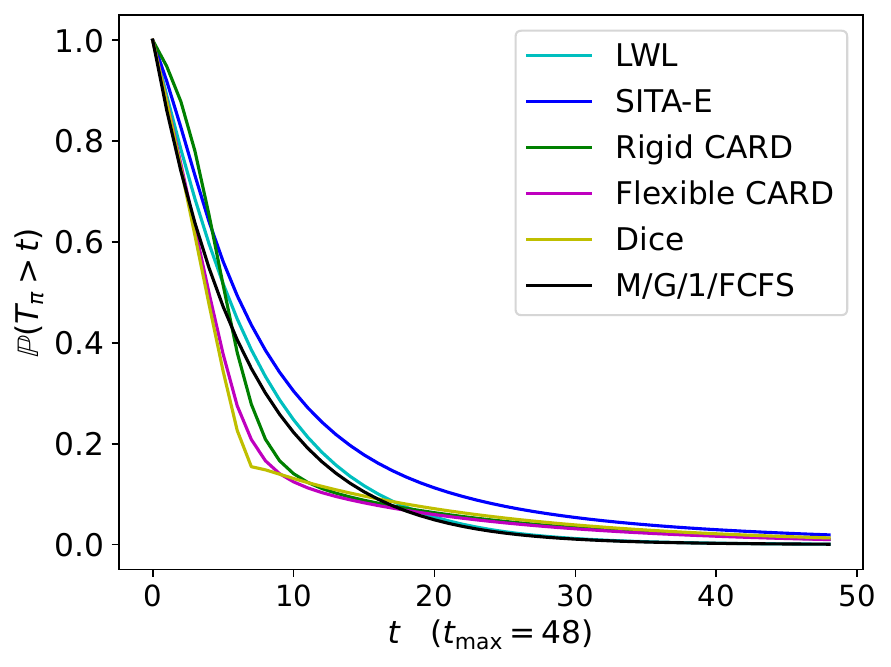}
    \caption{$2$ servers, $\mathsf{cv} = 1$, $\rho=0.85$}
\end{subfigure}\hfill
\begin{subfigure}[b]{0.31\textwidth}
    \includegraphics[width=\linewidth]{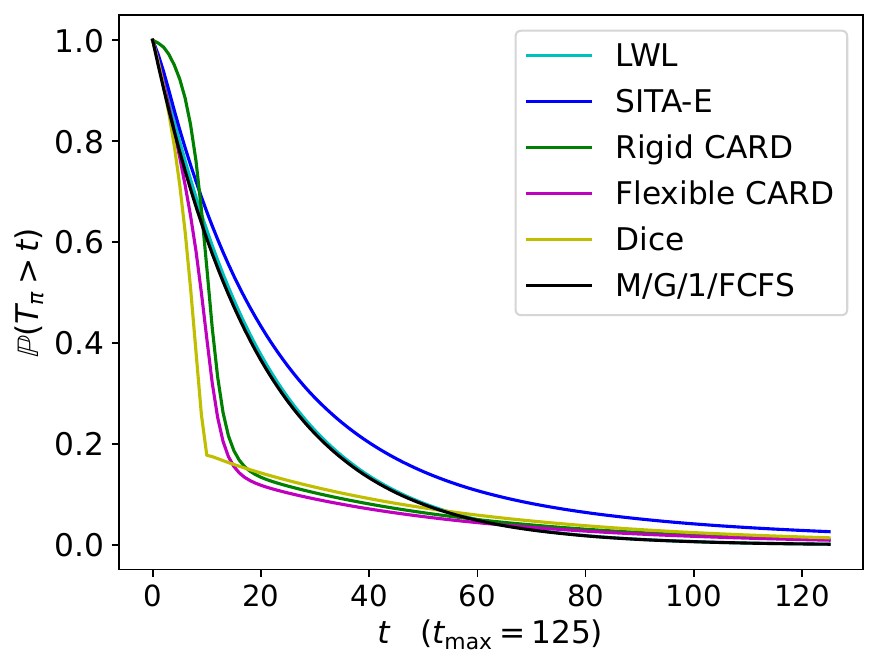}
    \caption{$2$ servers, $\mathsf{cv} = 1$, $\rho=0.95$}
\end{subfigure}
\begin{subfigure}[b]{0.31\textwidth}
    \includegraphics[width=\linewidth]{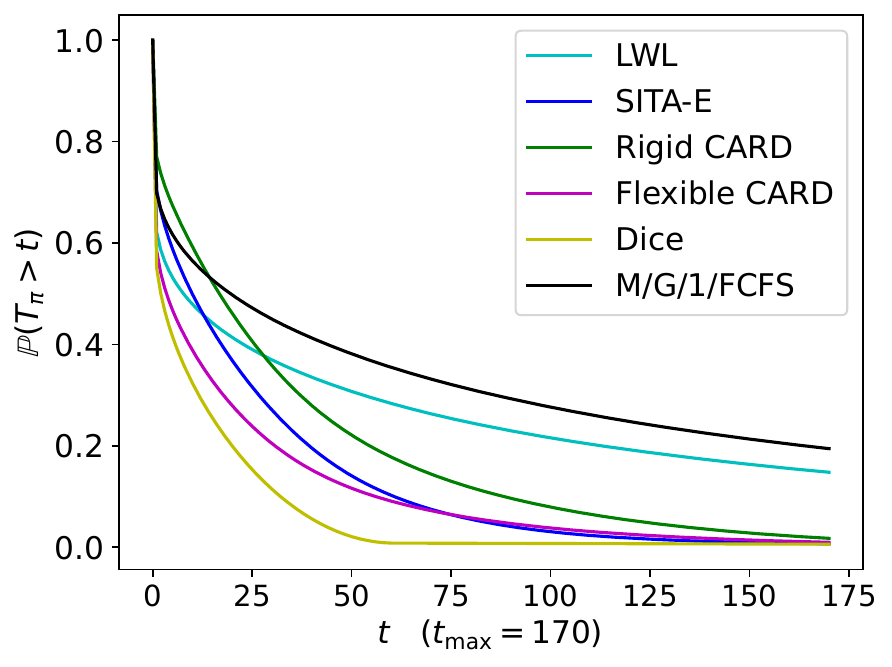}
    \caption{$2$ servers, $\mathsf{cv} = 100$, $\rho=0.7$}
\end{subfigure}\hfill
\begin{subfigure}[b]{0.31\textwidth}
    \includegraphics[width=\linewidth]{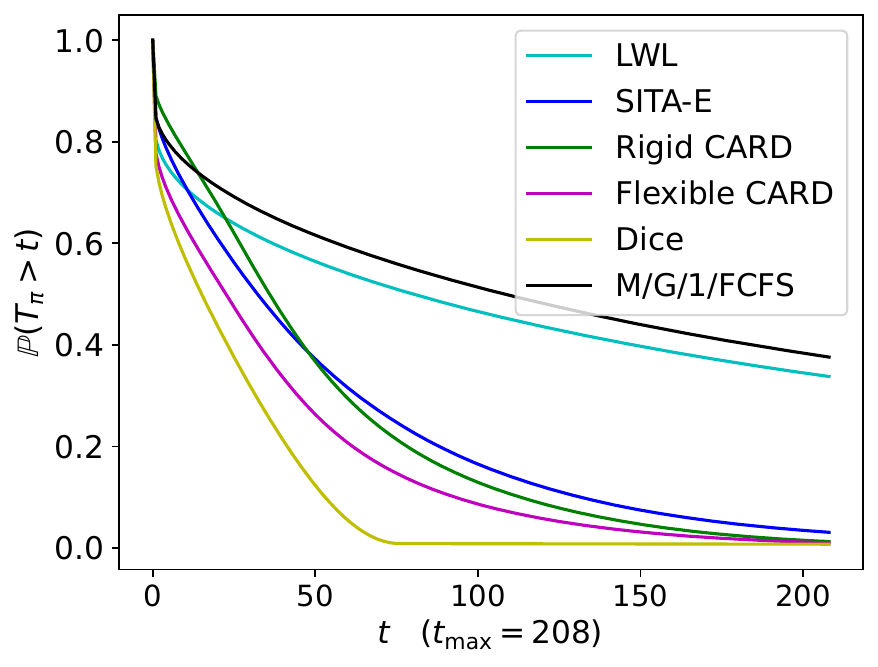}
    \caption{$2$ servers, $\mathsf{cv} = 100$, $\rho=0.85$}
\end{subfigure}\hfill
\begin{subfigure}[b]{0.31\textwidth}
    \includegraphics[width=\linewidth]{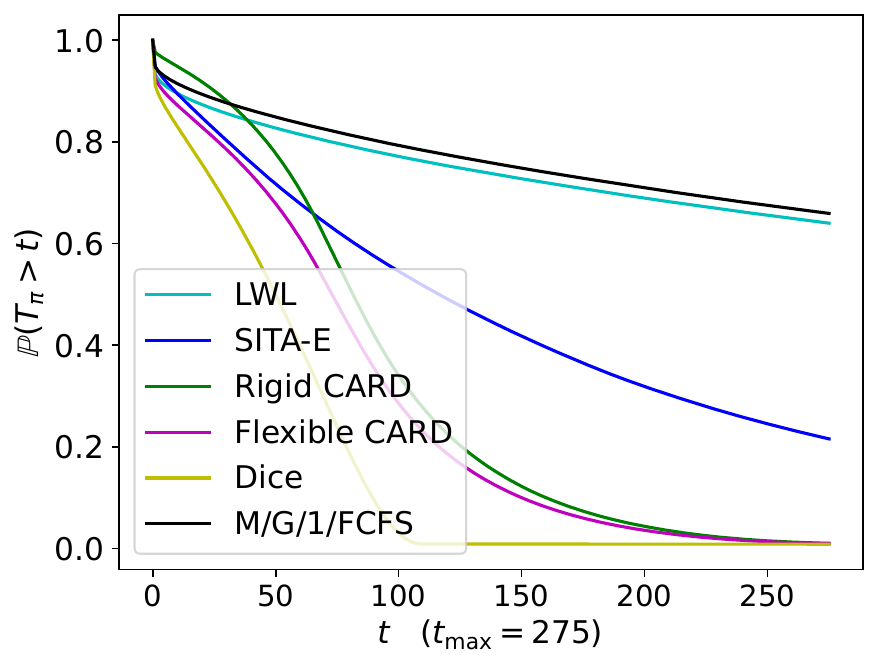}
    \caption{$2$ servers, $\mathsf{cv} = 100$, $\rho=0.95$}
\end{subfigure}
\caption{Response time tails for $n=2$ servers shown above. The tails are plotted in $[0,t_{\max}]$, where $t_{\max}$ is chosen so that $\P{T_{\text{Flexible CARD}}\leq t_{\max}}\approx0.99$.}
\label{fig:two-server-tail-comparison}
\end{figure}

\subsection{Comparing CARD to Dice}

Given the excellent performance of Dice, we feel Dice warrants a more in-depth discussion. We refer interested readers to \Cref{sec:dice}, where we discuss the following questions:
\* How complicated is Dice compared with CARD?
\* Why does Dice perform so well in simulations?
\* Is Dice heavy-traffic optimal?
\* Why is it hard to analyze Dice?
\*/
The main takeaway is that Dice may \emph{not} be heavy-traffic optimal, but it may be possible to modify Dice to make it heavy-traffic optimal. Dice remains a compelling option in practice that certainly deserves further study.

\section{Conclusion}

In this paper, we prove the first mean response time lower bound for FCFS servers. We design a new dispatching policy, called \emph{CARD (Controlled Asymmetry Reduces Delay)}, and show that it is heavy-traffic optimal, thus making CARD the first proven heavy-traffic optimal size- and state-aware dispatching policy. CARD can thus serve as a new benchmark policy for future work in dispatching or load-balancing for FCFS servers. Methodologically, our method of analyzing CARD using below-above cycles could be of independent interest, as it can be adapted to study other threshold-based policies.

Underlying our results is the insight that in the \emph{size-aware} dispatching setting, it is helpful to have a significant imbalance between the amounts of work at each server. This insight has been made multiple times throughout the size-aware dispatching literature (e.g. \citet{hyytia2012size, harchol1999choosing}). It is in contrast to the natural idea of always balancing the queues, which is helpful in size-oblivious dispatching \citep{zhou2018degree}.

In addition to minimizing mean response time, researchers today are also interested in tail performance. We conjecture that for job sizes with exponential tails and mean $1/\mu$, work under CARD asymptotically decays exponentially with rate $\frac{\mu-\lambda}{n}$ in an $n$-server system, which is worse (i.e. smaller) than the $\mu - \lambda$ decay rate of an M/M/1 under FCFS. An interesting follow-up question is how to balance the tradeoff between mean response time and response time decay rate, perhaps starting with the heavy-traffic regime. We leave this to future work.

\begin{acks}
We thank Rhonda Righter and Esa Hyyti\"a for helpful discussions about optimal dispatching and the Dice policy. We also thank Onno Boxma and Ivo Adan for helpful pointers on the G/M/1, which featured in a previous version of our below-period analysis.

Isaac Grosof was supported by \grantsponsor{nsf}{NSF}{https://www.nsf.gov} grant no. \grantnum{nsf}{CMMI-2307008}, and a Tennenbaum Postdoctoral Fellowship at the Georgia Institute of Technology School of Industrial and Systems Engineering. Ziv Scully conducted this research in part while visiting the Simons Institute for the Theory of Computing, and in part while a FODSI postdoc at Harvard and MIT. He was supported by \grantsponsor{nsf}{National Science Foundation}{https://www.nsf.gov} grant nos. \grantnum{nsf}{CMMI-2307008}, \grantnum{nsf}{DMS-2023528}, and \grantnum{nsf}{DMS-2022448}.
\end{acks}

\bibliographystyle{ACM-Reference-Format}
\bibliography{main}

\appendix

\section{Dice}
\label{sec:dice}
In this section, we give a brief introduction of Dice based on \citet{hyytia2022sequential} and answer some questions readers may be curious about.

\subsubsection{What is Dice?} Let $\mathbf{W}=(W_1,\ldots,W_n)$ be the workload vector. Dice has $n-1$ threshold parameters $\tau_1,\ldots,\tau_{n-1}$ to set. When a job of size $s$ arrives, Dice dispatches the job as follows:
\* Sort the workload vector in increasing order. We assume that $\tau_i$ is the threshold on $W_i$.
\* The dispatcher goes through $\mathbf{W}$ in increasing order and dispatches the job to the first server $i$ that satisfies $\tau_i-W_i>s$. If none of the first $n-1$ servers satisfy this condition, the job is dispatched to the last server. 
\*/

\subsubsection{How complicated is Dice compared with CARD?\nopunct} Dice is easier to implement than CARD. To implement Dice, one only needs to find decent threshold parameters $\tau_1,\ldots,\tau_{n-1}$. Successful implementation of CARD, however, requires tuning thresholds on job sizes (i.e. $m_-$ and $m_+$ for two-server CARD and $m_i$'s for multi-band CARD) in addition to the threshold parameters.

\subsubsection{Why does Dice perform so well in simulations?} Dice has the same spirit as CARD: it keeps one long queue populated mostly with large jobs and maintains short queues for small jobs to get through quickly, thereby increasing $\E{S_{\text{queue}}}$.

\subsubsection{Is Dice heavy-traffic optimal?} This is an open problem. We conjecture that Dice may \emph{not} be heavy-traffic optimal. Consider a two-server system. Under Dice, only jobs that fit into the gap between current work at the shorter server and the threshold get dispatched to the shorter server. However, this will not guarantee that as $\varepsilon \downarrow 0$, all jobs with size less than $m$ gets dispatched to the shorter server, as there is always some probability that the gap is too small for a job with size less than $m$.

\subsubsection{Why is it hard to analyze Dice?} Even with the tools we developed in this paper, Dice is more difficult to analyze than CARD. The main reason is that the dispatching policy under Dice changes continuously with the state of the shorter queue, which makes it challenging to directly apply our below-above cycle analysis. We expect that a refinement of our approach could work for Dice, but substantial extra work is needed.

\section{Deferred Proofs}
\label{sec:proofs}

\subsection{Suboptimality of LWL and SITA in heavy-traffic}
\label{sec:proofs:suboptimality}

In this section, we show that in a two-server system, neither LWL nor SITA is heavy-traffic optimal. Let
\[
K_\pi=\lim_{\varepsilon\downarrow0}\frac{\E{T_\pi}}{\E{W_{\mgone}}},
\]
be the \emph{heavy-traffic constant} of policy~$\pi$. In \cref{thm:k_lwl_and_k_sitae} below, we show that $K_\lwl$ and $K_\sitae$ are strictly greater than $K_\card$ as defined in~\cref{eq:Kcard}.

However, this does not finish the story for SITA, because while SITA-E is the version of SITA that splits the load equally, it is possible to improve SITA's performance by using an unequal load split. The version of SITA that uses the optimal load split is known as SITA-O. Surprisingly, SITA-O can have a significantly better heavy-traffic constant than SITA-E, even though there is very little flexibility in the amount of load each server can receive. In \cref{thm:k_sitao}, we sketch a computation of $K_\sitao$, showing that it, too, is strictly greater than $K_\card$.

For simplicity of computations, we focus on $n = 2$ servers, but the results should generalize to more servers. Similarly, while we continue to assume continuous job size distribution~$S$ (\cref{sec:model:basics}) for simplicity of defining SITA, the results should hold for any job size distribution for which $m$ is well-defined.\footnote{%
    When the distribution has an atom at~$m$, there are corner cases where $K_\sitao = K_\card$. One example is when $S \in \{a, b\}$ with probability~$1$ such that $a \P{S = a} = b \P{S = b}$ \citep{xie2023reducing}.}

\begin{theorem}
\label{thm:k_lwl_and_k_sitae}
In a system with $n=2$ servers and and a continuous job size distribution~$S$, we have $K_\lwl>K_\card$ and $K_\sitae \geq 2 K_\card$.
\end{theorem}
\begin{proof}
It is known that $K_\lwl=1$ (see e.g. \citep{whitt1993approximations}). To compute $K_\sitae$, we first note that under SITA-E, the system decouples to two independent M/G/1 queues. For any $\varepsilon\in(0,1)$, we have
\[
\label{eq:t_sitae}
\E{T_\sitae}&=2\gp*{\frac{\lambda\P{S<m}\E{S^2\given S<m}}{\varepsilon}\P{S<m}+\frac{\lambda\P{S\geqslant m}\E{S^2\given S\geqslant m}}{\varepsilon}\P{S\geqslant m}}\\
&\quad+ 2\E{S}\\
&=2\gp*{\frac{\lambda\E{S^2\I(S<m)}}{\varepsilon}\P{S<m}+\frac{\lambda\E{S^2\I(S\geqslant m)}}{\varepsilon}\P{S\geqslant m} + \E{S}}\\
&\eqnote{(a)}{\geq}\frac{2\lambda\E{S^2}}{\varepsilon}\P{S\geqslant m}+2\E{S}\\
&=4\P{S\geqslant m}\E{W_\mgone}+2\E{S},
\]
where (a) follows from the fact that $\P{S<m}\geqslant\P{S\geqslant m}$. Looking at the $\epsilon \downarrow 0$ limit, we have $K_\sitae \geq 2 K_{\card}$.
\end{proof}

\begin{theorem}
\label{thm:k_sitao}
In a system with $n=2$ servers and a continuous job size distribution~$S$, we have $K_\sitao>K_\card$.
\end{theorem}
\begin{proof}[Proof sketch]
SITA-O works like SITA-E, except instead of using size threshold~$m$ to split jobs between the servers, it uses a different size threshold~$m'$. The key insight is that in the $\epsilon \downarrow 0$ limit, we must have $m' - m \leq O(\epsilon)$, beause otherwise we would overload one of the servers. This means, roughly speaking, that SITA-O can affect the denominators in \cref{eq:t_sitae}, but it cannot significantly affect the numerators. Specifically, there exists $x \in (-1, 1)$ such that
\[
\E{T_\sitao}&=2\gp*{\frac{\lambda \P{S<m} \E{S^2\given S<m}}{\epsilon (1 - x)} \P{S<m}  +\frac{\lambda \P{S\geq m} \E{S^2 \given S\geq m}}{\epsilon (1 + x)} \P{S\geq m} \pm O(1)}\\
&=\frac{2\lambda}{\epsilon} \gp*{\frac{\E{S^2 \I(S < m)} \P{S<m}}{1 - x} + \frac{\E{S^2 \I(S \geq m)} \P{S\geq m}}{1 + x}} \pm O(1).
\]
Optimizing over the value of $x$ yields
\[
K_\sitao &= \frac{2}{\E{S^2}}\gp*{\sqrt{\E{S^2 \I(S < m)} \P{S < m}} + \sqrt{\E{S^2 \I(S \geq m)} \P{S \geq m}}}^2 \\
&\eqnote{(a)}{\geq} \frac{2 \P{S \geq m}}{\E{S^2}}\gp*{\sqrt{\E{S^2 \I(S < m)}} + \sqrt{\E{S^2 \I(S \geq m)}}}^2 \\
&= K_\card \gp*{1 + \frac{2 \sqrt{\E{S^2 \I(S < m)} \E{S^2 \I(S \geq m)}}}{\E{S^2}}},
\]
where (a) follows from the fact that $\P{S<m}\geqslant\P{S\geqslant m}$.
\end{proof}

\subsection{Stability}
\label{sec:proofs:stability}

As outlined in \cref{sec:stability}, we begin by showing that the short server is stable for any threshold $c\geq0$. Our main tool is a continuous-time Foster-Lyapunov theorem developed in \Citet{meyn1993stability3}. A key component of the theorem is the \emph{infinitesimal generators} for Markov processes. Let $X(t)$ be a Markov process, its infinitesimal generator, $\mathscr{A}$, is the operator defined by
\[
\mathscr{A}V(x)=\lim_{t\downarrow0}\frac{\E{V(X(t))\given X(0)=x}-V(x)}{t}.
\]
The domain of $\mathscr{A}$ is all functions~$V$ for which the limit on the right exists for all~$x$ in the state space. Since work at the short server, $W_s(t)$, is a Markov process, for a function $V$ with left derivative, we may explicitly derive the infinitesimal generator of $W_s(t)$ under CARD: 
\[
\mathscr{A}V(w_s)&=-\tfrac{1}{2}V'(w_s)\,\I(w_s>0)+\lambda p_s\E_{S_s}{V(w_s+S_s)-V(w_s)}\\
&\quad+\I(w_s\leq c)\lambda p_m \E_{S_m}{V(w_s+S_m)-V(w_s)},
\]
where $p_s=\P{S\leq m}$, $p_m=\P{m_-<S<m_+}$, $\E_{S_s}{\cdot}$ is the expectation over the distribution of small jobs (i.e. $S\given S\leq m_-$) and $\E_{S_m}{\cdot}$ is the expectation over the distribution of medium jobs (i.e. $S\given m_-<S<m_+$), $\I(\cdot)$ is the indicator function, and $V'$ is the left derivative.

We now present the continuous-time Foster-Lyapunov theorem \Citep[Theorem~4.4]{meyn1993stability3} below for easy reference.

\begin{theorem}\label{thm:foster_lyapunov}
Suppose that a Markov process $\mathbf{\Phi}$ is a non-explosive right process. If there exists constants $c,d>0$, a function $f\geq1$, a closed petite set $C$, and a function $V\geq0$ that is bounded on $C$ such that for all $x\in O_m$ and $m\in\mathbb{Z}$,
\[\mathscr{A}_m V(x) \leq - \eta f(x) + d \I_C(x),\]
then $\mathbf{\Phi}$ is positive Harris recurrent.
\end{theorem}

Here, $O_m$ is a family of precompact sets that increases to the entire state space as $m\to\infty$ and $\mathscr{A}_m$ is the generator for the truncated process restricted to $O_m$. This restriction is in place mainly to handle possibly explosive processes. Our process $\mathbf{W}(t)$ is not explosive. More importantly, the Lyapunov function $V$ we consider in \cref{thm:stability_short} is increasing and differentiable. It follows that $\mathscr{A}_mV(x)\leq\mathscr{A}V(x)$ for all $x$. It therefore suffices for us to apply theorem \cref{thm:foster_lyapunov} with $\mathscr{A}V(x)$ instead.

\restate*\ref{thm:stability_short}
\begin{proof}
\label{pf:stability_short}
We first check the preconditions of \cref{thm:foster_lyapunov} hold for $W_s$ under CARD for any $\epsilon>0$. $W_s$ is obviously non-explosive. Let $V(W_s)=W_s$, $C=\{W_s: W_s\leq c\}$, and $f(w_s)\equiv1$. Then we have
\begin{align*}
\mathscr{A}V(W_s)=\beta \I(W_s\leq c)-\alpha \I(W_s>c)
\end{align*}
Since $\alpha>0$ for any $\epsilon>0$, positive Harris recurrence of $W_s(t)$ follows from \cref{thm:foster_lyapunov} if $C$ is a closed petite set. We now check this.

It follows from \citet[Theorem~3.1]{meyn1993stability3} that $W_s(t)$ is non-evanescent. By \citet[Theorem~4.2]{meyn1993stability2}, the $K_a$-chain of $W_s(t)$ is an irreducible $T$-process with everywhere nontrivial continuous component. By \citet[Theorem~4.1(i)]{meyn1993stability2}, $C$ is a petite set. 

Given positive recurrence of $W_s(t)$ and \citet[Theorem~4.2]{meyn1993stability2}, we conclude from \citet[Theorem~3.2]{meyn1994stability} that $W_s(t)$ is also ergodic.
\end{proof}

\restate*\ref{thm:ssc_short}
\begin{proof}
\label{pf:ssc_short}
Define $V(w_s)=(c-w_s)^+$ and fix some $\theta>0$. Since $W_s$ has a stationary distribution, we can apply the rate conservation law \citep{miyazawa1994rate} to $e^{\theta V(W_s)}$, which yields
\[
\frac{\theta}{n}\E_{\pi}{e^{\theta V(W_s)}\I(V(W_s) < c)}+\lambda_{s.m}\E_{\pi,S_{s,m}}{e^{\theta V(W_s+)}-e^{\theta V(W_s)}}=0.\]
Here, $\pi$ is the stationary distribution of $W_s$ and $\lambda_{s,m}$ is the arrival rate into a short server from small and medium jobs, and $V(W_s+)$ is the value of $V(W_s)$ immediately after a job arrival of size $S_{s,m}$. Rearranging yields
\[
    \frac{\theta}{n}\E_{\pi}{e^{\theta V(W_s)}}
    + \lambda_{s,m}\E_{\pi,S_{s,m}}{e^{\theta V(W_s+)}-e^{\theta V(W_s)}}= \frac{\theta}{n}\E_{\pi}{e^{\theta V(W_s)} \I(V(W_s) = c)}.
\]
We drop the RHS and work with
\[
\label{eq:RCL_emptiness_bound}
\frac{\theta}{n}\E_{\pi}{e^{\theta V(W_s)}}+\lambda_{s,m}\E_{\pi,S_{s,m}}{e^{\theta V(W_s+)}-e^{\theta V(W_s)}}\geq0.\]
We first analyze the second term on the LHS. Conditioning on a given state $W_s$, we have
\[\MoveEqLeft\E_{S_{s,m}}*{e^{\theta V(W_s+)}-e^{\theta V(W_s)}\given W_s}\\
&\eqnote{(a)}{=} \E_{S_{s,m}}*{e^{\theta (V(W_s)-S_{s,m})^+}-e^{\theta V(W_s)}\given W_s}\\
&= \E_{S_{s,m}}*{e^{\theta (V(W_s)-S_{s,m})^+}-e^{\theta(V(W_s)-S_{s,m})}+e^{\theta(V(W_s)-S_{s,m})}-e^{\theta V(W_s)}\given W_s}\\
&\eqnote{(b)}{=} \E_{S_{s,m}}*{e^{\theta (V(W_s)-S_{s,m})^+}-e^{\theta(V(W_s)-S_{s,m})}\given W_s}+e^{\theta V(W_s)}\E_{S_{s,m}}*{(e^{-\theta S_{s,m}}-1)}\\
&\eqnote{(c)}{=} \E_{S_{s,m}}*{(1-e^{\theta(V(W_s)-S_{s,m})})\,\I(V(W_s)<S_{s,m})\given W_s}+e^{\theta V(W_s)}(\widetilde{S_{s,m}}(\theta)-1),\label{eq:RCL_emptiness_bound_2}
\]
where
\*[(a)] follows from $V(W_s+)=(c-W_s-S_{s,m})^+=(V(W_s)-S_{s,m})^+$,
\* follows from the independence of the arriving job size and $V(W_s)$ for any given $W_s\leq c$, and
\* follows from the definition of LST and the fact that
\[
e^{\theta (V(W_s)-S_{s,m})^+}-e^{\theta(V(W_s)-S_{s,m})}=
\begin{cases}
0,\quad &\text{if }V(W_s)\geq S_{s,m}\\
1-e^{\theta(V(W_s)-S_{s,m})}, &\text{if }V(W_s)<S_{s,m}.
\end{cases}
\]
\*/

Taking expectation over $\pi$ on both sides of \eqref{eq:RCL_emptiness_bound_2} and substituting into \eqref{eq:RCL_emptiness_bound} gives
\[
\frac{\theta}{n}\E_{\pi}{e^{\theta V(W_s)}}
&\geq \lambda_{s,m}\E_{\pi}*{e^{\theta V(W_s)}}(1-\widetilde{S_{s,m}}(\theta))-\lambda_{s,m}\E_{\pi,S_{s,m}}*{(1-e^{\theta(V(W_s)-S_{s,m})})\,\I(V(W_s)<S_{s,m})} \\
&\eqnote{(d)}{=} \theta\gp*{\frac{n-1}{n}+(n-1)\beta}\E_\pi*{e^{\theta V(W_s)}}\widetilde{(S_{s,m})_e}(\theta) \\*
&\quad - \theta\gp*{\frac{n-1}{n}+(n-1)\beta}\E_{\pi,S_{s,m}}*{\frac{1-e^{\theta(V(W_s)-S_{s,m})}}{\theta\E{S_{s,m}}}\I(V(W_s)<S_{s,m})},\mkern24mu
\label{eq:RCL_emptiness_bound_3}
\]
where (d) follows from $\frac{n-1}{n}+(n-1)\beta=\lambda_{s,m}\E{S_{s,m}}$, which is the load of the short and medium jobs, and the fact that
\[\widetilde{(S_{s,m})_e}(\theta)=\frac{1-\widetilde{S_{s,m}}(\theta)}{\theta\E{S_{s,m}}},\]
which holds for a general job size distribution (with or without a density function). See e.g. \Citet{ross1995stochastic}.

Since
\[\left(1-e^{-\theta\left(S_{s,m}-V(W_s)\right)}\right)\I(V(W_s)<S_{s,m})\leq1-e^{-\theta S_{s,m}},\]
we have
\[\E_{\pi,S_{s,m}}*{\frac{1-e^{\theta(V(W_s)-S_{s,m})}}{\theta\E{S_{s,m}}}\I(V(W_s)<S_{s,m})}\leq\widetilde{(S_{s,m})_e}(\theta).\]
Since $\theta\geq0$ is chosen so that $\widetilde{(S_{s,m})_e}(\theta)>\frac{1}{n(n-1)\beta+n-1}$, we have $\gp*{\tfrac{n-1}{n}+(n-1)\beta}\widetilde{(S_{s,m})_e}(\theta)-\tfrac{1}{n}>0$. Thus, we rearrange \eqref{eq:RCL_emptiness_bound_3} to obtain
\[\E_\pi{e^{\theta V(W_s)}}\leq\frac{\gp*{n(n-1)\beta+n-1}\widetilde{(S_{s,m})_e}(\theta)}{\gp*{n(n-1)\beta+n-1}\widetilde{(S_{s,m})_e}(\theta)-1}.\]
Markov's inequality then gives
    \[
        \P{V(W_s) < x} \leq \frac{\gp*{n(n-1)\beta+n-1}\widetilde{(S_{s,m})_e}(\theta)}{\gp*{n(n-1)\beta+n-1}\widetilde{(S_{s,m})_e}(\theta)-1}e^{-\theta x},
    \]
and the lemma follows.
\end{proof}

\restate*\ref{thm:stability}
\begin{proof}
\label{pf:stability}
Part (a) is a corollary of \cref{thm:idleness_short}. For part (b), we first establish the result for $n=2$ servers, then show how it generalizes to $n>2$ servers. For $n=2$, we denote the state as $\mathbf{W}(t)=(W_s(t),W_\ell(t))$.

To establish~(b), we first apply \Citet[Theorem~1]{foss2012stability} to the pre-jump chain $\{\mathbf{W}(T_n-)\}$, where $T_n$ is the arrival time of the $n$th job. Conditions A1-A3 in Theorem 1 are fulfilled by \cref{thm:stability_short}. Define
\[L_2(w_\ell)=\lambda w_\ell,\qquad f(w_s)=-\tfrac{1}{2}+\rho_\ell+\rho_m\I(w_s>c),\qquad h(x)=\tfrac{1}{2}e^{-x}.\]
We now verify conditions B1 and B2 are met with above choices of $L_2$, $f$, and $h$.

\paragraph{Condition B1:}
\[\sup_{w_s,w_\ell}\E{|L_2(W_\ell(T_1-))-L_2(w_\ell)|\given W_\ell(0-)=w_\ell,W_s(0-)=w_s,\text{Arrival at time $0$}}\leq1.\]

\paragraph{Condition B2:}
Let $\pi_s$ be the stationary distribution of $W_s(t)$. By PASTA, $\pi_s$ is also the stationary distribution of $\{W_s(T_n-)\}$. We have
\[\E_{\pi_s}{f(W_s)}&=-\tfrac{1}{2}+\rho_\ell+\rho_m\P{W_s>c}\\
&\stackrel{(a)}{\leq} -\tfrac{1}{2}+\rho_\ell+\rho_m\frac{(\rho_m+\rho_s)-\frac{1}{2}+\frac{\delta}{2}}{\rho_m}\\
&=-1+\rho+\frac{\delta}{2}=-\epsilon+\frac{\delta}{2}<0,
\]
where (a) comes from \cref{thm:prob_above_below} and PASTA. We then compute\footnote{%
    The conditional probabilities below are a slight abuse of notation. They should be understood as referring to the probability measure induced by the pre-jump Markov chain starting from state $(w_s, w_\ell)$.}
\[
    \MoveEqLeft
    \E{L_2(W_\ell(T_1-))-L_2(w_\ell))\given W_\ell(0-)=w_\ell,W_s(0-)=w_s,\text{Arrival at time $0$}}\\
    &= \lambda \E*{\int_0^{T_1-}-\tfrac{1}{2}\I(W_\ell(t)>0)\d{t}\given W_\ell(0-)=w_\ell,W_s(0-)=w_s,\text{Arrival at time $0$}}
    \\* &\quad +\rho_\ell+\rho_m\I(w_s>c)\\
    &= f(w_s)+\frac{\lambda}{2}\E*{\gp*{T_1-(w_\ell+S)}^+\given W_\ell(0-)=w_\ell,W_s(0-)=w_s,\text{Arrival at time $0$}}\\
    &\leq f(w_s)+\frac{\lambda}{2}\E*{\gp*{T_1-w_\ell}^+}=f(w_s)+\frac{\lambda}{2}\frac{1}{\lambda}e^{-\lambda w_\ell}=f(w_s)+h(L_2(w_\ell)).
\]
It now follows from \citet[Theorem~1]{foss2012stability} that the embedded pre-jump chain $\{\mathbf{W}(T_n-)\}$ is positive Harris recurrent.

Since $\{\mathbf{W}(T_n-)\}$ is positive Harris recurrent and easily seen to be $\{(0, 0)\}$-irreducible, the expected number of steps until returning to $(0,0)$ is finite from any starting state. The time between steps is exponentially distributed with mean~$1/\lambda$, so we conclude from Wald's equation that the expected return time of the original process $\mathbf{W}(t)$ to state $(0,0)$ is also finite. Positive Harris recurrence of $\mathbf{W}(t)$ immediately follows.
We now generalize the above proof to $n>2$ servers. To begin with, we define a vector-valued process $\mathbf{W}_{\text{short servers}}(t)=(W_{s_1}(t),\ldots,W_{s_{n-1}}(t))$. Under multiserver CARD, $W_{\text{short servers}}(t)$ has the following properties:
\* $\mathbf{W}_{\text{short servers}}(t)$ is a Markov process of its own and is Harris ergodic.
\* Since stationary distribution of the short servers are i.i.d., the stationary distribution of $\mathbf{W}_{\text{short servers}}$ is the product of stationary distributions of the short servers in isolation.
\*/
With these two properties in hand, the argument for $n=2$ servers as presented above works for $n>2$ servers with the same functions $h$, $L_2$, and the following~$f$:
\[
f(\mathbf{w}_{\text{short servers}})=-\frac{1}{n}+\rho_\ell+\frac{\rho_m}{n-1}\sum_{i=1}^{n-1}\I(w_{s_i}>c).
\qedhere
\]
\end{proof}

\subsection{Response Time Analysis}

\restate*\ref{thm:work_short}
\begin{proof}
    \label{pf:work_short}
    As stated in the proof sketch, $(W_s - c \given W_s > c)$ has the same distribution as an M/G/1 with vacations.
    \* The job size distribution is $S_s = (S \given S < m_-)$. In particular, using the fact that $S_s$ is stochastically dominated by~$m_+$, one can show that $(S_s)_e$ is stochastically dominated by a uniform distribution on $[0, m_+]$.\footnote{%
        It is not in general true that $S$ being dominated by $R$ implies $S_e$ is dominated by $R_e$. This is specific to the case that $R$ is a deterministic constant.}
    \* The load is $1 - 2 \alpha$, and so the slackness is $2 \alpha$.
    \** The reason we use $2 \alpha$ instead of $\alpha$ is because the server operates at speed~$\frac{1}{2}$. By ``doubling the clock speed'', the server speed becomes~$1$, and the distribution of $(W_s - c \given W_s > c)$ is unaffected. This makes it easy to apply standard results about the M/G/1 with vacations.
    \* Let $U$ denote the vacation length distribution. It is hard to characterize exactly, but because $W_s - c \leq m_+$ at the start of an above period, $U_e$ is stochastically dominated by a uniform distribution on $[0, m_+]$.
    \*/
    The desired bounds follow from the work decomposition formula for the M/G/1 with vacations \citep{fuhrmann1985stochastic}. Specifically, for an M/G/1 with vacations, we can write its steady-state work $W_{\mathrm{M/G/1/vac}}$ as an independent sum of random variables with distributions $W_\mgone$ and $U_e$. This means
    \[
        \E{W_s - c \given W_s > c} = \E{W_{\mathrm{M/G/1/vac}}} &= \E{W_\mgone} + \E{U_e}, \\
        \E{(W_s - c)^2 \given W_s > c} = \E{W_{\mathrm{M/G/1/vac}}^2} &= \E{W_\mgone^2} + 2 \E{W_\mgone} \E{U_e} + \E{U_e^2}.
    \]
    Applying the PK formula with the relevant parameters, we obtain
    \[
        \E{W_\mgone} &= \frac{(1 - 2 \alpha) \E{(S_s)_e}}{2 \alpha}
        \leq \frac{(1 - 2\alpha) m_+}{4 \alpha}, \\
        \E{W_\mgone^2} &\leq \frac{(3 - 4 \alpha (2 - \alpha)) m_+^2}{24 \alpha ^2}.
    \]
    The result then follows from $\E{U_e} \leq \frac{m_+}{2}$ and $\E{U_e^2} \leq \frac{m_+^3}{3}$.
\end{proof}

\restate*\ref{thm:work_long_below+above}
\begin{proof}
    \label{pf:work_long_below+above}
    The proof is very similar to that of \cref{thm:Wl_work_change}, so we give only the key steps. 
    Applying the Palm inversion formula \citep{baccelli2002elements} to $W_\ell \I(W_s>c)$ gives
    \[
        \E{W_\ell\I(W_s> c)}
        = \frac{1}{\E{A+B}} \E_c^0*{\int_{B}^{A+B} W_\ell(t) \d{t}},
    \]
    where we can start the integral at~$B$ and remove the indicator because $W_s(t) > c$ exactly during above periods, which corresponds to $t \in [B, A+B)$. Expanding this using \cref{eq:Wl_process} and noting the independence of $W_\ell(0)$ from the below-above cycle, we obtain
    \[
        \MoveEqLeft
        \E{W_\ell\I(W_s> c)} - \frac{\E{A}}{\E{A+B}} \E_c^0{W_\ell(0)} \\
        &= \frac{1}{\E{A+B}}\E_c^0*{\int_B^{A+B} \max\curlgp*{-\Delta_\ell(0,t)+\Sigma_\ell^m(0,t)+\Sigma_\ell^\ell(0,t), -W_\ell(0)}\d{t}}.
    \]
    Applying \cref{eq:prob_above_below_renewal} to the left-hand side, we see it suffices to give bounds on the right-hand side. The same reasoning as the proof \cref{thm:Wl_work_change} yields
    \[
        \vgp[\big]{\E{W_\ell\I(W_s> c)} - q_A \E_c^0{W_\ell(0)}}
        \leq \frac{1}{\E{A+B}}\E_c^0*{\int_B^{A+B} t \d{t}}
        = \frac{\E{(A + B)^2 - B^2}}{2 \E{A + B}}.
    \]
    The result then follows from a computation similar to the end of the proof of \cref{thm:Wl_work_change}.
\end{proof}

\restate*\ref{thm:rt_upper_explicit}
\begin{proof}
    \label{pf:rt_upper_explicit}
    Consider a tagged job arriving to the system. Recall from \cref{eq:rt_CARD} that
    \[
        \E{T_{\card}} - 2 \E{S}
        &\leq 2 (p_s + p_m) \E{W_s} + 2 p_m \E{W_\ell \I(W_s > c)} + 2 p_\ell \E{W_\ell}.
    \]
    We now bound the work expectations and probabilities in the last line.
    \* \cref{thm:work_short} implies $\E{W_s} \leq c + q_A \E{W_s - c \given W_s > c} \leq c + \frac{q_A m_+}{\alpha}$.
    \* \cref{thm:Wl_work_change, thm:work_long_below+above} imply, after some simplification,
    \[
        \E{W_\ell \I(W_s > c)}
        \leq q_A \E{W_\ell} + \frac{q_A m_+}{\alpha^2} + \frac{4 q_A q_B c}{\beta} + \frac{\sqrt{2 q_A q_B m_+ c}}{\alpha \sqrt{\beta}}.
    \]
    \* \cref{thm:bound_on_W} bounds $\E{W_\ell}$.
    \*/
    From these bounds and some simplification, using facts like $p_s + p_m + p_\ell = 1$ and $\alpha \leq 1$, we obtain
    \[
        \E{T_{\card}} - 2 \E{S}
        &\leq 2 (p_\ell + q_A) \gp*{1 + \frac{\delta}{\epsilon}} \E{W_\mgone}
            + \frac{4 q_A m_+}{\alpha^2}
            + \frac{m_+ \sqrt{q_A}}{\alpha \sqrt{\epsilon}}
            \\* &\quad
            + 6 c + \frac{8 q_A q_B c}{\beta} + \frac{2 \sqrt{2 q_A q_B m_+ c}}{\alpha \sqrt{\beta}}
            + \frac{8 c \sqrt{\delta}}{\alpha^2 \beta \epsilon}.
    \]
    We now use \cref{thm:prob_above_below} to express as much as possible on the right-hand side in terms of $\alpha$, $\beta$, $\delta$, and~$\epsilon$. After some simplification, including using the preconditions of the theorem, we obtain
    \[
        \E{T_{\card}} - 2 \E{S}
        &\leq 2 \gp*{p_\ell + \frac{2 \beta}{\alpha + \beta}} \gp*{1 + \frac{\delta}{\epsilon}} \E{W_\mgone}
            + \frac{8 m_+ \beta}{\alpha^2 (\alpha + \beta)}
            + \frac{m_+ \sqrt{2 \beta}}{\alpha \sqrt{\epsilon (\alpha + \beta)}}
            \\* &\quad
            + \gp*{\frac{12 m_+}{\beta} + \frac{32 m_+ \alpha}{\beta (\alpha + \beta)^2}} \log\frac{3}{2 \beta \delta}
            + \frac{4 m_+}{\alpha + \beta} \sqrt{\frac{2}{\alpha \beta} \log\frac{3}{2 \beta \delta}}
            + \frac{16 m_+ \sqrt{\delta}}{\alpha^2 \beta^2 \epsilon} \log\frac{3}{2 \beta \delta}.
    \]
    Finally, we observe that $p_\ell = \P{S > m_+} \leq \P{S > m}$ and simplify further.
\end{proof}

\subsection{Extension to Any Number of Servers}
\label{sec:proofs:multi}

\restate*\ref{thm:rt_upper_heavy}
\begin{proof}[Proof for $n \geq 2$ servers]
Fix a short server $s_i$. Notice that under multi-server CARD, $W_{s_i}$ are i.i.d. Thus, the analysis applies to any short server. Let $A$ and $B$ be the above and below periods of $W_{s_i}$. Using a similar proof as that of \cref{thm:Wl_work_change}, we obtain
    \[
        \vgp[\big]{\E{W_\ell}-\E_c^0{W_\ell(0)}}
        \leq \gp*{\sqrt{q_A \E{A_e}} + \sqrt{q_B \E{B_e}}}^2
        \leq \frac{q_A m_+}{2 \alpha^2} + \frac{4 q_B c}{\beta}.
    \]
For any short server $i$, we have
    \[
        q_A=\P{W_{s_i}>c} \leq \frac{\beta+\frac{1}{n}\delta}{\alpha + \beta} \leq \frac{2 \beta}{\alpha + \beta} \qquad
        \text{and} \qquad
        q_B=\P{W_{s_i}\leq c} \leq \frac{\alpha}{\alpha + \beta}.
    \]
The proof is similar to that of \cref{thm:prob_above_below}. Note that $q_A$ and $q_B$ are the same for all short servers because $W_{s_i}$ are i.i.d. in steady state.
\cref{thm:mean_excess_bound_below, thm:mean_excess_bound_above} follow from the same arguments as the two-server case. We would like to obtain a counterpart of \cref{thm:bound_on_W}. To this end, we use a multi-server version of \cref{thm:work_decomp}. Note that we have
\[
I W_\all&=\frac{1}{n}\sum_{i=1}^{n-1}\I(W_{s_i}=0)W_\ell+\frac{1}{n}\sum_{i=1}^{n-1}\I(W_\ell=0)W_{s_i}+\frac{1}{n}\sum_{k\neq j}\I(W_{s_k}=0)W_{s_j}
\]
We bound these three terms separately.
\[
\frac{1}{n}\sum_{i=1}^{n-1}\E{\I(W_{s_i}=0)W_\ell}&\eqnote{(a)}{=}\frac{n-1}{n}\E{W_\ell\I(W_{s_1}=0)}\\
&\leq\frac{n-1}{n}\sqgp*{\delta \gp*{\E{W_\all} + \frac{q_A m_+}{4 \alpha^2} + \frac{4 q_B c}{\beta}} + \frac{2 c \sqrt{2 \delta}}{\beta}},\\
\frac{1}{n}\sum_{i=1}^{n-1}\E{\I(W_\ell=0)W_{s_i}}&\eqnote{(b)}{=}\frac{n-1}{n}\E{\I(W_\ell=0)W_{s_1}}\leq\frac{n-1}{n}\gp*{n\epsilon c+\frac{m_+ \sqrt{q_A n\epsilon}}{2\sqrt{2}\alpha}},\\
\frac{1}{n}\sum_{k\neq j}\E{\I(W_{s_k}=0)W_{s_j}}&\eqnote{(c)}{\leq}\frac{(n-1)(n-2)}{n}\P{W_{s_k}=0}\E{W_j}\leq(n-1)\gp*{c+\frac{m_+q_A}{\alpha}}\delta,
\]
where (a), (b), and (c) all follow from the fact that $W_{s_1},\ldots, W_{s_{n-1}}$ are i.i.d. in steady state. Proof of the other bounds are similar to their counterparts in \cref{thm:IW_analysis, thm:bound_on_W}. \cref{thm:work_decomp} gives
\[
    \E{W_\ell}
    \leq \E{W_\all}
    &\leq \gp[\bigg]{1 + (n-1)\frac{\delta}{\epsilon}} \E{W_\mgone} \\* &\quad
        + n(n-1) c +\sqrt{n}(n-1) \frac{m_+ \sqrt{q_A}}{2\sqrt{2} \alpha \sqrt{\epsilon}}
        + \frac{8}{n}\frac{ c \sqrt{\delta}}{\alpha^2 \beta \epsilon}+\frac{n(n-1)}{\epsilon}\gp*{c+\frac{m_+q_A}{\alpha}}\delta.
\]
from \cref{thm:prob_above_below}. Since $W_{s_1},\ldots, W_{s_{n-1}}$ are i.i.d. in steady state, we have, by PASTA,
\[
\P{\text{A medium job joins a short server queue}}=\P{W_{s_1}\leq c}=q_B
\]
Therefore, using an argument similar to that for \cref{thm:rt_upper_explicit}, we have
\[
        \E{T_{\card}} - n \E{S}
        &\leq n (p_\ell + q_A) \gp*{1 +(n-1) \frac{\delta}{\epsilon}} \E{W_\mgone}
            + \frac{2n q_A m_+}{\alpha^2}
            \\* &\quad
            + n(n-1)\sqrt{n}\frac{m_+ \sqrt{q_A}}{2\sqrt{2} \alpha \sqrt{\epsilon}}
            + \gp*{n^2(n-1)+n} c + \frac{4n q_A q_B c}{\beta} + \frac{n \sqrt{2 q_A q_B m_+ c}}{\alpha \sqrt{\beta}}
            \\* &\quad
            + \frac{8 c \sqrt{\delta}}{\alpha^2 \beta \epsilon}+n^2(n-1)\gp*{c+\frac{m_+q_A}{\alpha}}\frac{\delta}{\epsilon}.
    \]
This can be further expanded using the bounds for $q_A$ and $q_B$, as well as the expression of $c$.
\[
\E{T_{\card}} - n \E{S}
        &\leq n \gp*{p_\ell + \frac{2\beta}{\beta+\alpha}} \gp*{1 +(n-1) \frac{\delta}{\epsilon}}\E{W_\mgone}+\frac{4n\beta m_+}{\alpha^2(\alpha+\beta)}+n(n-1)\frac{m_+\sqrt{n\beta}}{2\alpha\sqrt{\epsilon(\alpha+\beta)}}
        \\*&\quad
        +\gp*{\frac{n(n-1)(n^2(n-1)+n)m_+}{\beta}+\frac{8n^2(n-1)m_+\alpha}{\beta(\alpha+\beta)^2}}\log\frac{n+1}{n\beta\delta}\\*&\quad
        +\frac{2n\sqrt{n(n-1)}m_+}{\alpha+\beta}\sqrt{\frac{1}{\alpha\beta}\log\frac{n+1}{n\beta\delta}}+\frac{8n(n-1)m_+\sqrt{\delta}}{\alpha^2\beta^2\epsilon}\log\frac{n+1}{n\beta\delta}
        \\*&\quad
        +\underbrace{n^2(n-1)\gp*{\frac{n(n-1)m_+}{\beta}\log\frac{n+1}{n\beta\delta}+\frac{2m_+\beta}{\alpha(\alpha+\beta)}}\frac{\delta}{\epsilon}}_{\mathcal{T}}
\]
At this point, we note that the upper bound for $\E{T_{\card}} - n \E{S}$ is, after letting $n=2$, the same as that in \cref{thm:rt_upper_explicit}, except for $\mathcal{T}$. Thus, setting
    \[
        \alpha = \Theta(1), \qquad
        \beta = \Theta\gp[\bigg]{\epsilon^{1/3} \gp[\bigg]{\log\frac{1}{\epsilon}}^{2/3}}, \qquad
        \text{and} \qquad
        c = \frac{n(n-1)m_+}{\beta} \log \frac{n+1}{n \beta \delta},
    \]
and noting that $\mathcal{T}\to0$ as $\epsilon\downarrow0$, we conclude that the bound yields the same heavy-traffic scaling as that in \cref{thm:rt_upper_explicit}.  Finally, we note that $K_\card$ emerges because
\[
\lim_{\epsilon\downarrow0}np_\ell=n\P{S>m}=K_\card.
\qedhere
\]
\end{proof}

\section{Additional Simulations}
Our additional simulations applies flexible CARD with three parameters to $n=10$ servers. We simulate 40 trials for each data point, with $10^7$ arrivals per trial for $\mathsf{cv} = 1$ and $\mathsf{cv} = 10$ and $3\times10^7$ job arrivals per trial for $\mathsf{cv} = 100$. We show 95\% confidence intervals when wider than the marker size. 

\Cref{fig:n-server-flexible} show that, for $n=10$ servers, flexible CARD has decent performance when the coefficient of variation is small. However, for large coefficients of variation, flexible CARD does not perform well, even if we use LWL to dispatch small and medium jobs among the short servers. Specifically, when cv=10, flexible CARD deviates from Dice, although still better than LWL and SITA-E. When cv=100, flexible CARD performs worse than SITA-E at high loads. The unsatisfactory performance of flexible CARD for $n=10$ servers motivates us to design multi-band CARD.

\begin{figure}
\begin{subfigure}[b]{0.31\textwidth}
    \includegraphics[width=\linewidth]{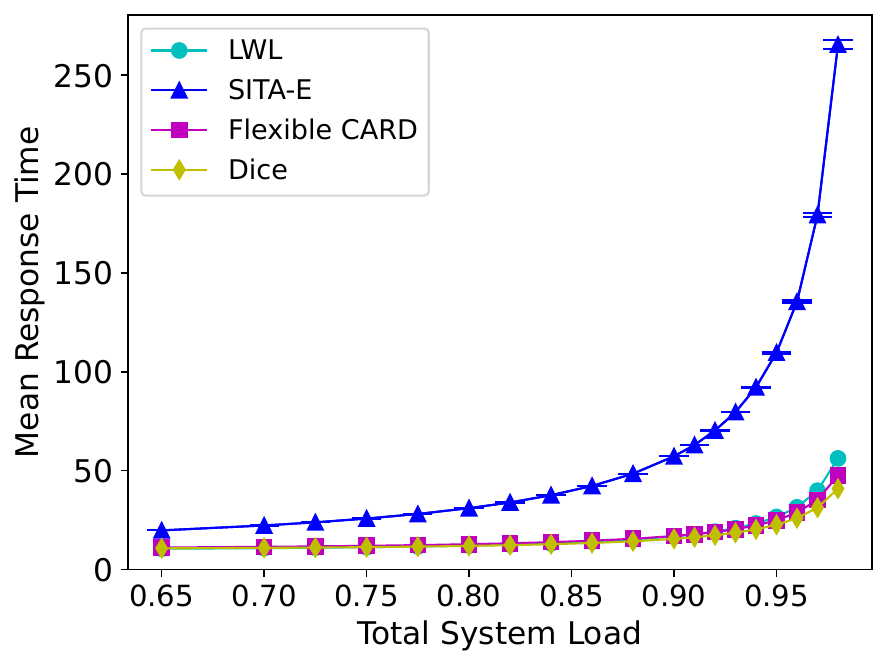}
    \includegraphics[width=\linewidth]{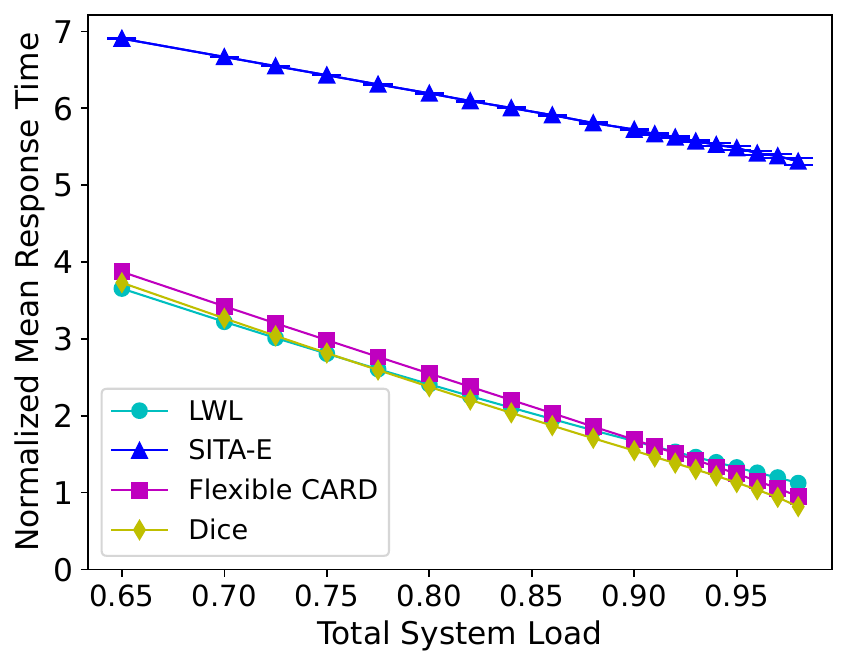}
    \caption{$10$ servers, $\mathsf{cv} = 1$}
\end{subfigure}\hfill
\begin{subfigure}[b]{0.31\textwidth}
    \includegraphics[width=\linewidth]{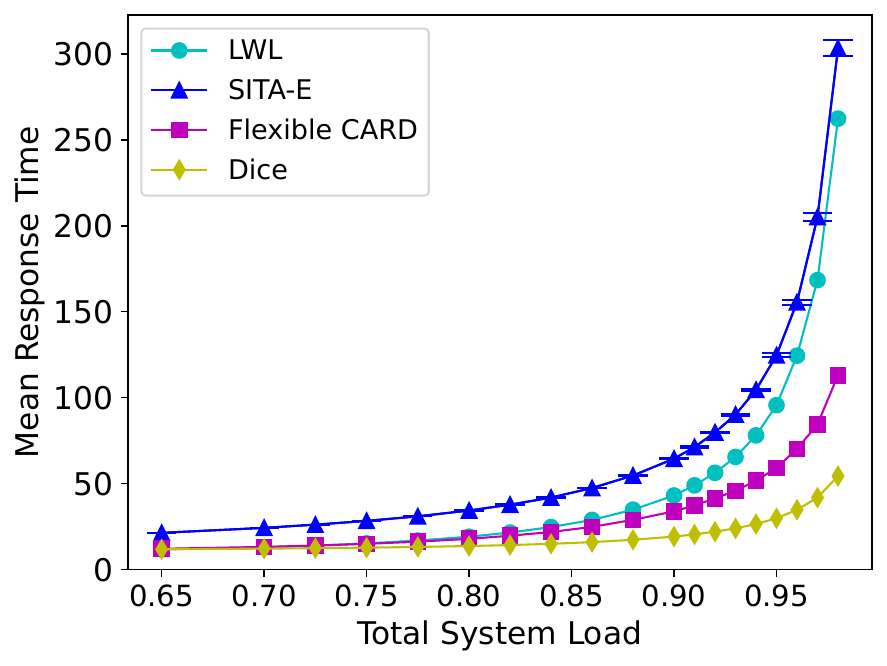}
    \includegraphics[width=\linewidth]{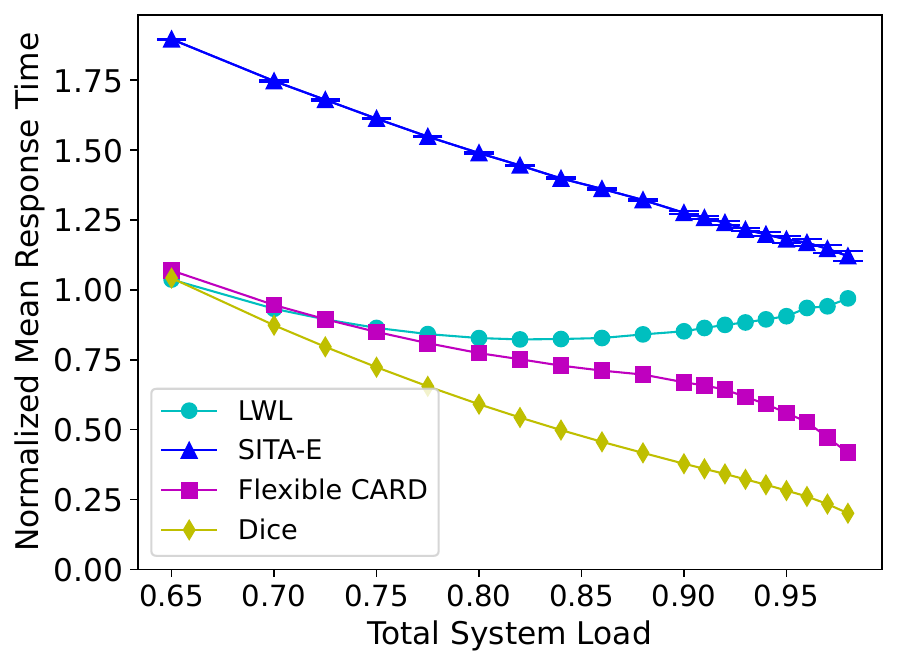}
    \caption{$10$ servers, $\mathsf{cv} = 10$}
\end{subfigure}\hfill
\begin{subfigure}[b]{0.31\textwidth}
    \includegraphics[width=\linewidth]{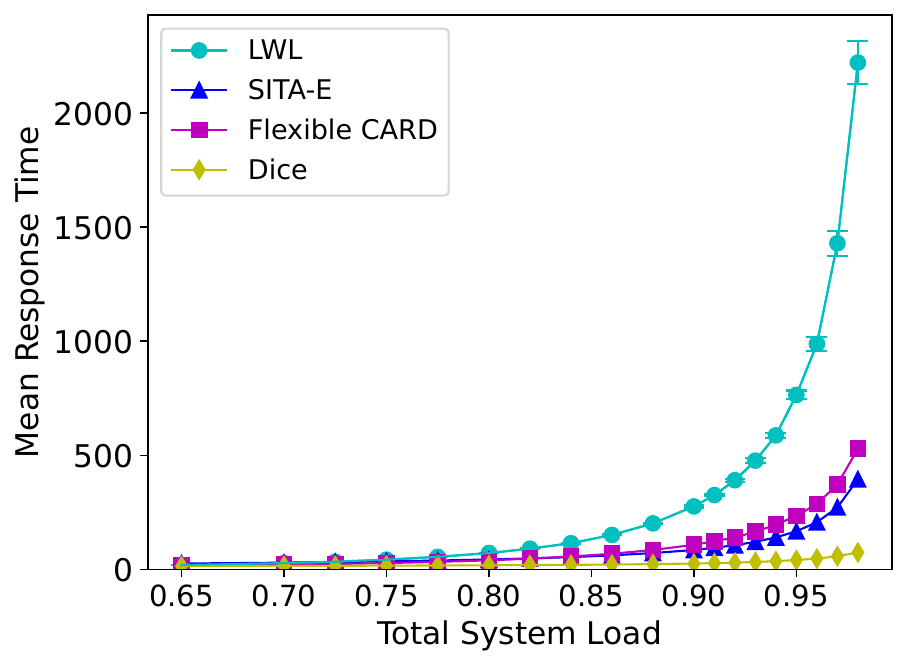}
    \includegraphics[width=\linewidth]{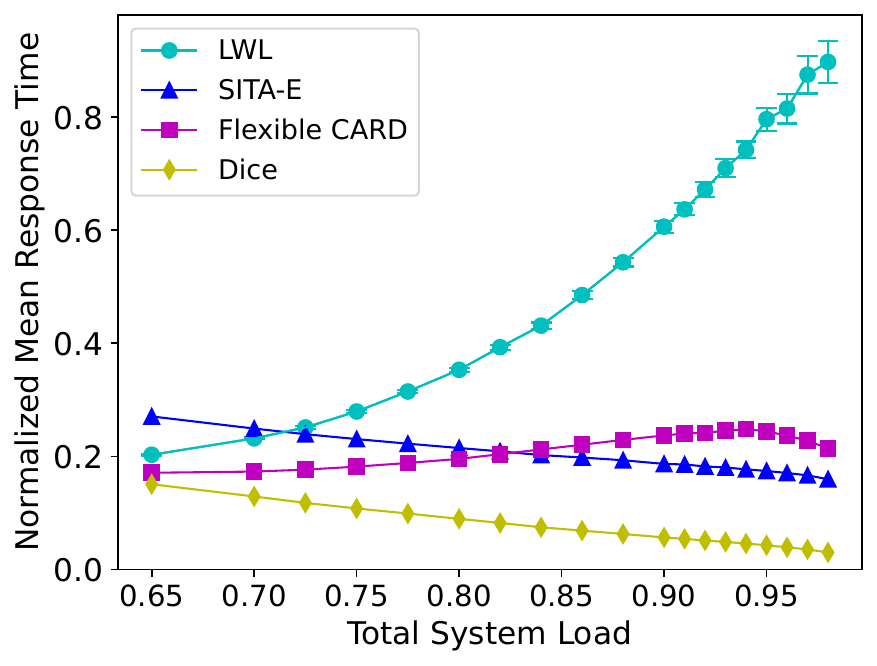}
    \caption{$10$ servers, $\mathsf{cv} = 100$}
\end{subfigure}
\caption{Plots for the mean response times under the aforementioned policies for $n = 10$ servers. On the top row are plots of mean response times of the policies. On the bottom row are plots for mean response times normalized by the mean response time of a resource-pooled M/G/1 queue. We use LWL, instead of random, dispatching to short servers when a small or medium job arrives.}
\label{fig:n-server-flexible}
\end{figure}

\newpage
\received{October 2023}
\received[revised]{January 2024}
\received[accepted]{January 2024}

\end{document}